%% file: main.tex
\documentclass[12pt, reqno, a4paper,oneside]{amsart}

\rmfamily{\fontsize{12pt}{\baselineskip}\selectfont}

\usepackage{graphicx}
\usepackage{tikz}
\usepackage{color}
\usepackage[colorlinks=true, allcolors=blue]{hyperref}
\usepackage{xcolor}
\usepackage[top=2.3cm,bottom=1.8cm,left=1.8cm,right=1.8cm]{geometry}
\usepackage{amsfonts, amsmath, amssymb, amsbsy, amsthm}
\usepackage{environments}
\usepackage{mathrsfs}
\usepackage{subcaption}
\usepackage{bm}
\usepackage[normalem]{ulem}
\usepackage{listings}
\usepackage{algpseudocode}
\usepackage{algorithm, algorithmicx}
\numberwithin{theorem}{section}
\numberwithin{equation}{section}
\definecolor{yscol}{HTML}{6622AA}

\definecolor{cobalt}{rgb}{0.0, 0.28, 0.67}

\makeatletter
\renewcommand\subsubsection{\@startsection{subsubsection}{3}{\z@}%
                                     {-3.25ex\@plus -1ex \@minus -.2ex}%
                                     {1.5ex \@plus .2ex}%
                                     {\normalfont\normalsize\bfseries}}
\makeatother

\definecolor{darkgreen}{rgb}{0.0, 0.5, 0.0} % RGB for dark green

\hypersetup{
    colorlinks,
    linkcolor=green, % color of internal links (change box color with linkbordercolor)
    citecolor=green, % color of links to bibliography
    filecolor=magenta, % color of file links
    urlcolor=cyan % color of external links
}

\title[Surrogate Models for Vibrational Entropy]{Surrogate Models for Vibrational Entropy Based on a Spatial Decomposition}

\author{Tina Torabi}
\address{Tina Torabi\\
Department of Mathematics\\
University of British Columbia\\
1984 Mathematics Road\\
Vancouver, British Columbia\\
Canada
}
\email{torabit@math.ubc.ca}

\author{Yangshuai Wang}
\address{Yangshuai Wang\\
Department of Mathematics\\
University of British Columbia\\
1984 Mathematics Road\\
Vancouver, British Columbia\\
Canada
}
\email{yswang2021@math.ubc.ca}

\author{Christoph Ortner}
\address{Christoph Ortner\\
Department of Mathematics\\
University of British Columbia\\
1984 Mathematics Road\\
Vancouver, British Columbia\\
Canada
}
\email{ortner@math.ubc.ca}

\date{\today}

\begin{document}

\maketitle

\input{notation}

\begin{abstract}
The temperature-dependent behavior of defect densities within a crystalline structure is intricately linked to the phenomenon of vibrational entropy. Traditional methods for evaluating vibrational entropy are computationally intensive, limiting their practical utility. Building on~\cite{julian2016} we show that total entropy can be decomposed into atomic site contributions and rigorously estimate the locality of site entropy. This analysis suggests that vibrational entropy can be effectively predicted using a surrogate model for site entropy. We employ machine learning to develop such a surrogate models employing the Atomic Cluster Expansion model. 
We supplement our rigorous analysis with an empirical convergence study. 
In addition we demonstrate the performance of our method for predicting vibrational formation entropy and attempt frequency of the transition rates, on point defects such as vacancies and interstitials.
% The results pave the way for predicting various entropy-dependent physical quantities, including but not limited to dynamical properties like free energy and diffusion coefficients.
\end{abstract}

% Some references for us: \cite{bartok2018machine, lysogorskiy2021performant, zuo2020performance} 
% https://ai-atoms.github.io/milady-docs/index.html 
% Anharmonic thermo-elasticity of tungsten from accelerated Bayesian adaptive biasing force calculations with data-driven force fields

\section{Introduction}
\label{sec:intro}

% ...\cys{basically speaking, the introduction should contain several paragraphs: background; motivations and key problems; literature reviews; the primary focus of our work; outlook;}

Entropy plays a crucial role in understanding and characterizing the behavior exhibited by condensed matter systems, particularly in the context of crystalline materials undergoing the aging process. The thermodynamic and kinetic characteristics of defects within these materials significantly influence their evolution, resulting in a diverse range of morphologies distinguished by variations in size, character, and density. The role of entropy becomes particularly prominent in indicating how the stability of defect populations changes in response to temperature fluctuations.

Vibrational entropy, arising from the thermal movement of the lattice structure, emerges as a critical component in this context. At temperatures below the melting point, the vibrational aspect takes precedence, becoming the primary contributor to the total entropy. Consequently, assessing vibrational entropy provides valuable insights into the behavior of system in various practical scenarios. However, a fundamental challenge in evaluating harmonic vibrational entropy lies in the computational costs. The classical approach involves calculating the dynamical matrix, a process requiring on the order of $O(N^2)$ calculations, and subsequent diagonalization, which demands $O(N^3)$ calculations, with $N$ representing the total number of atoms in the system. As a result, the computation of entropy on large length scales becomes computationally intractable, presenting a computational challenge that needs to be addressed.

To overcome this computational challenge, recent work by Lapointe et al.~\cite{PhysRevMaterials.6.113803} has employed machine learning surrogate models utilizing atomic environment descriptors. These models offer precise predictions for the vibrational formation entropy of point defects, achieving significant computational efficiency with calculations on the order of $O(N)$. The versatility of these data-driven models is further enhanced by the incorporation of higher-dimensional descriptor spaces, establishing a non-linear relationship between descriptors and observables \cite{PhysRevMaterials.4.063802}. This advancement enables a more comprehensive representation of physical phenomena, such as defect formation and migration, essential for accurate rate approximations in transition state theory. However, it is noteworthy that these studies implicitly assumed that the total vibrational entropy can be expressed as a local contribution, and the resulting site entropy is local. Limited research has delved into a rigorous analysis of vibrational entropy. To the best knowledge of the authors, the only existing work in this domain is~\cite{julian2016}. In this study, the formation free energy and transition rates between stable configurations, assessed using periodic supercell approximations, converge as the cell size increases. While this work implicitly utilized the locality of site entropy, it cannot be directly applied to justify machine learning surrogates.

The purpose of the present paper is to undertake a rigorous analysis of the locality of site entropy, utilizing insights presented from~\cite{julian2016}, where total entropy was deconstructed into contributions from individual atomic sites. This analysis establishes the groundwork for constructing entropy based on the local atomic environment. Building upon this foundation, we employ machine learning surrogate models with atomic environment descriptors to precisely and efficiently predict the vibrational formation entropy of crystalline defects. To that end we utilize the Atomic Cluster Expansion (ACE) method, widely applied in the field of machine-learning interatomic potentials (MLIPs) for atomistic simulations. Numerical experiments are conducted, focusing on point defects such as vacancies and interstitials. The robustness of our approach is demonstrated by accurately predicting vibrational entropy and the attempt frequency for the transition rates governing point defect migration, a critical element in transition state theory rate approximations. These results lay the groundwork for predicting various entropy-dependent physical quantities, including but not limited to dynamical properties like free energy and diffusion coefficients, which can potentially be integrated into kinetic Monte Carlo methods~\cite{voter2007introduction}, where transition rates from state to state are pivotal.

This paper primarily concentrates on relatively simple settings (point defects) to provide a comprehensive theoretical and practical analysis of key concepts, but a wide variety of extensions to more complex classes of crystalline defects, and to other classes of materials are conceivable. 

% Future research will extend the method to more complex scenarios such as dislocations. However, extrapolating to grain boundary structures poses fundamental challenges beyond the current heuristic analysis, necessitating the development of novel approaches.

\subsection*{Outline}

In Section \ref{sec:pre}, we establish the foundation for our theories by offering background information and preliminaries. This includes rigorous derivations of both total and site entropy. Section \ref{sec:th_rs} presents our primary findings (Theroem~\ref{thm:local}), focusing on the locality of site entropies in an infinite lattice $\Lambda$, supported by numerical results validating our main theorem. Moving on to Section \ref{sec:num_rs}, we delve into the background of the surrogate models utilized for entropy fitting and subsequently present our entropy and corresponding attempt frequency fitting results. The Appendices contain a compilation of auxiliary results and proofs essential for the understanding and validation of the preceding sections.

\subsection*{Notation}
Let \(X\) be a (semi-)Hilbert space and let its dual be represented by  \(X^*\). The duality pairing is denoted by \(\langle \cdot, \cdot \rangle\). The space of bounded linear operators mapping from \(X\) to another (semi-)Hilbert space \(Y\) is expressed as \(\mathcal{L}(X, Y)\). 

% For \(\E \in C^2(X)\), the first variation is denoted by \(\delta \E(x) \in X^*\), and \(\langle \delta \E(x), v \rangle\) (with \(v \in X\)) signifies the directional derivative. Additionally, the second variation is denoted by \(\delta^2 \E(x) \in \mathcal{L}(X, X^*)\) and may colloquially be referred to as the Hessian.
For $\mathcal{E} \in C^2(X)$, the first and second variations are denoted by
$\<\delta \mathcal{E}(u), v\>$ and $\<\delta^2 \mathcal{E}(u) v, w\>$ for $u,v,w\in X$, i.e., 
\begin{equation}
    \begin{aligned}
    \langle\delta\mathcal{E}(u),v\rangle&:=\lim_{t\rightarrow0}t^{-1}\big(\mathcal{E}(u+tv)-\mathcal{E}(u)\big),\\
    \langle\delta^{2}\mathcal{E}(u)v,w\rangle&:=\lim_{t\rightarrow0}t^{-2}\langle\delta\mathcal{E}(u+tw)-\delta\mathcal{E}(u),v\rangle. \nonumber
    \end{aligned}
\end{equation}
It is easy to see that $\delta \E(x) \in X^*$ and $\delta^2 \E(x) \in \mathcal{L}(X, X^*)$.

Assuming $\Lambda$ is a countable index set (often a Bravais lattice, $\L = \mA \Z^d$, with $\mA \in \R^{d\times d}$ being non-singular), we define $\ell^2(\Lambda;\R^m)=\{u : \L \to \R^m : \sum_{\ell \in \L} |u|^2 < \infty\}$. If the range is evident from the context, this can be abbreviated to $\ell^2(\Lambda)$ or simply $\ell^2$.

For $A \in \mathcal{L}\big(\ell^2(\Lambda;\mathbb{R}^m), \ell^2(\Lambda;\mathbb{R}^m)\big)$, we define its element $A_{\ell i n j} := \big\langle A (\delta_\ell e_i), \delta_n e_j \big\rangle_{\ell^2(\Lambda;\mathbb{R}^d)}$, where $\ell,n \in \Lambda$ and $i,j \in \{1,\ldots,m\}$. Additionally, we denote $A_{\ell n} = (A_{\ell i n j})_{ij} \in \mathbb{R}^{m \times m}$ as the matrix blocks corresponding to atomic sites. The identity is represented by $(I_{\ell^2(\Lambda;\mathbb{R}^m)})_{\ell i n j} := \delta_{\ell n} \delta_{ij}$.

For $j\in\N$, ${\bm{g}}\in (\R^d)^A$, and $V \in C^j\big((\R^d)^A\big)$, we define the notation
\begin{eqnarray*}
	V_{,{\bm \rho}}\big({\bm g}\big) :=
	\frac{\partial^j V\big({\bm g}\big)}
	{\partial {\bm g}_{\rho_1}\cdots\partial{\bm g}_{\rho_j}}
	\qquad{\rm for}\quad{\bm \rho}=(\rho_1, \ldots, \rho_j)\in A^{j}.
\end{eqnarray*}
 
The symbols $C, c$ denote generic positive constants that may change from one line
of an estimate to the next. When estimating rates of decay or convergence, $C, c$
will always remain independent of approximation parameters such as the system size, the configuration of the lattice and the test functions. The dependence of $C, c$ will be clear from the context or stated explicitly. To further simplify notation we will often write $\lesssim$ to mean $\leq C$.

\section{Background and Preliminaries}
\label{sec:pre}
Although our approach to constructing surrogates for vibration entropy is very general, our rigorous analysis relies on the setting of crystalline solids, potentially with defects. 
To motivate the formulation of our main results in this context, we will review the framework developed in \cite{2021-defectexpansion, bcshap2016, ortner2023framework} and in particular the renormalisation analysis of vibrational formation entropy~\cite{julian2016}. For the sake of simplicity of presentation, we will skip over some technical details but fill these gaps in Appendix~\ref{sec: appendix_prelem}.

Let $d\in\{2,3\}$ be the dimension of the system. A homogeneous crystal reference configuration is given by the Bravais lattice $\Lambda^{\rm hom}=\mathsf{A}\Z^d$, for some non-singular matrix $\mathsf{A} \in \mathbb{R}^{d \times d}$. We admit only single-species Bravais lattices. There are no conceptual obstacles to generalising our work to multi-lattices \cite{olson2023elastic}, however, the technical details become more involved. The reference configuration with defects is a set $\Lambda \subset \R^d$. The mismatch between $\Lambda$ and $\Lambda^{\rm hom}$ represents possible defected configurations. We assume that the defect cores are localized, that is, there exists $R^{\rm def}>0$, such that $\Lambda \backslash B_{R^{\rm def}} = \Lambda^{\rm hom} \backslash B_{R^{\rm def}}$.

The displacement of the infinite lattice $\Lambda$ is a map $u \colon \Lambda \to \R^m$. For $\ell, \rho \in \Lambda$, we denote discrete gradients (or, differences) by $D_\rho u (\ell) := u(\ell + \rho) - u(\ell)$. Higher order differences are denoted by $D_{\boldsymbol{\rho}} = D_{\rho_1} \cdots D_{\rho_j}$ for a $\boldsymbol{\rho} = (\rho_1, ..., \rho_j) \in \Lambda^j$. For a subset $\mathcal{R} \subset \Lambda-\ell$, we
define $Du(\ell) := D_{\mathcal{R}_{\ell}} u(\ell) := \big(D_{\rho} u(\ell)\big)_{\rho\in\mathcal{R}_{\ell}}$. We assume throughout that $\mathcal{R}_{\ell}$ is {\em finite} for each site $\ell\in\Lambda$. An extension of our analysis to infinite interaction range is not conceptually difficult but involves additional technical and notational complexities~\cite{chen19}. 

We consider the site potential to be a collection of mappings $V_{\ell}:(\R^m)^{\mathcal{R}_\ell}\rightarrow\R$, which represent the energy distributed to each atomic site. We make the following assumption on regularity and symmetry:
$V_\ell \in C^K\big((\R^m)^{\mathcal{R}_\ell}\big)$ for some $K$ and $V_\ell$ is homogeneous outside the defect region, namely, $V_\ell = V$ and $\mathcal{R}_\ell = \mathcal{R}$ for $\ell \in \Lambda \setminus B_{R^{\rm def}}$. Furthermore, $V$ and $\mathcal{R}$ have the following point symmetry: $\mathcal{R} = -\mathcal{R}$, it spans the lattice $\textrm{span}_\Z \mathcal{R}= \Lambda$, and $V\big(\{-A_{-\rho}\}_{\rho\in\mathcal{R}}\big) = V(A)$. We refer to~\cite[Section 2.3]{chen19} for a detailed discussion of those assumptions.

\subsection{Supercell model}
\label{sec:sub:form}
%
% \cys{Next, in order to ..., we consider the supercell model ...} \cys{can we think about a stronger reason for this?} 
Following \cite{julian2016} we initially consider the periodic setting and then reference the established existence of the thermodynamic limit and proceed to analyze the infinite lattice. To that end, let $\mathsf{B}=(b_1, \ldots, b_d) \in \R^{d \times d}$ be invertible such that $b_i \in \mathsf{A}\Z^d$. For a sufficiently large $N \in \N$ with $R^{\rm def} \ll N$, we denote
\[
\L_N:=\L\cap\mathsf{B}(-N/2, N/2]^d \quad \textrm{and} \quad \L_N^{\per}:= \bigcup_{\alpha \in N\Z^d}(\mathsf{B}\alpha + \L_N),
\]
where $\L_N$ is the periodic computational domain and $\L_N^{\per}$ is the periodically repeated domain. 

We define the space of periodic displacements to be
\begin{align*}
\mathcal{W}^{\text{per}}_N := \{u:\L_N^{\per} \rightarrow \R^m~|~u(\ell+\mathsf{B}\alpha) = u(\ell) ~\textrm{for}~\alpha \in N\Z^d\}.
\end{align*}
An equilibrium defect geometry configuration is determined by
\begin{align}\label{eq:opt}
    \bar{u}_N &\in \arg\min\{\mathcal{E}_N(u)~|~u \in \mathcal{W}^{\text{per}}_N\}, \nonumber \\
    \text{where}~\mathcal{E}_N&(u) := \sum_{\ell \in \Lambda_N} V_{\ell}\big(Du(\ell)\big) \quad \text{for}~u \in \mathcal{W}^{\text{per}}_N.
\end{align}
Analogously, for future reference, we introduce the energy functional for the homogeneous (defect-free) supercell as
\begin{equation}\label{eq:Ehom}
   \mathcal{E}^{\text{hom}}_N(u) := \sum_{\ell \in \Lambda_N} V\big(Du(\ell)\big) \quad \text{ for } u \in \mathcal{W}^{\text{per}}_N.
\end{equation}

\subsection{Vibrational entropy}
\label{subsec:FFE}

The vibrational entropy is closely related to the {\it formation free energy}. This is used in, for example, phase diagrams \cite{FULTZ2010247}, diffusion coefficients \cite{10.1063/5.0007178}, equilibrium concentration of defects \cite{PhysRev.93.265}.

In the harmonic approximation model (thus incorporating only vibrational entropy into the
model) we approximate a nonlinear potential energy landscape (cf.~\eqref{eq:opt} and \eqref{eq:Ehom}) by a quadratic
expansion about an energy minimizer of interest,
\begin{align}
   \E^{\rm hom}_N(w) & \approx
            {\textstyle \frac{1}{2}} \langle H^{\rm hom}_N w,w \rangle,
            \quad \text{and} \\
   \E_N(\bar{u}_N + w) & \approx
      \E_N(\bar{u}_N) + {\textstyle \frac{1}{2}} \langle H_N(\bar{u}_N) w, w \rangle,
\end{align} 
where we use the fact that  $\delta\mathcal{E}^{{\rm hom}}_N({\bf 0})$ and $\delta\mathcal{E}_N(\bar{u}_N)$ vanish. Additionally, we denote $H_N(u) := \delta^2\mathcal{E}_N(u), \; H^{\text{hom}}_N(u) := \delta^2 \mathcal{E}^{\text{hom}}_N(u),$ and $H^{\text{hom}}_N := H^{\text{hom}}_N({\bf 0})$ for the Hessians of systems.

% Knowledge of the free energy is vital for forecasting the stability of various states of materials. Numerous physical and chemical phenomena, such as phase diagrams, partition coefficients, equilibrium concentration of defects, and so on, are closely related to the free energy of the system. Through statistical mechanics, the changes in free-energy can be described using averages of atomic arrangements in the system \cite{fe2007}.

% % Consider a crystalline solids with an embedded defect. In this section, we'll discuss a how to calculate the vibrational entropy of such a system. 

% The nonlinear potential energy landscapes in the harmonic approximation model are approximated by the corresponding quadratic expansions around the the energy minima of interest, thus we can express the homogeneous and defectuous systems' energy landscapes as 

The free energy of a system is intimately tied to the partition function. The harmonic approximation of the partition function is given by

% \cco{this is not a chain of equalities, please reformulate, maybe just skip the second group?} \cco{also - to make it clearer, should we write $Z_{\rm h}$  and $Z_{\rm h}^{\rm hom}$?}

% \begin{equation}
% \begin{aligned}
% Z_{\rm h} &= \int e^{-\beta \E_N(u)} \mathrm{d} u 
% = C_{\beta,N} \cdot e^{-\beta \E_N(\bar{u}_N)} \cdot \big(\text{det}^{+}(H_N)\big)^{-1/2},
% \end{aligned}
% \end{equation}

\begin{equation}
\begin{aligned}
Z_{\rm h} &= e^{-\beta \E_N(\bar{u}_N)} \int e^{ -\beta \frac{1}{2} \langle H_N(\bar{u}_N) w, w \rangle } \mathrm{d} w \\
&= C_{\beta,N} \cdot e^{-\beta \E_N(\bar{u}_N)} \cdot \big(\text{det}^{+}(H_N)\big)^{-1/2},
\end{aligned}
\end{equation}
where the constant $C_{\beta, N} = (2\pi/ \beta)^{((2N)^d-1)m/2}$ with $\beta:=1/k_{B}T$ the inverse temperature. Here $\det^+(H_N) := \prod_j \lambda_j$ with $\lambda_j$ representing the positive eigenvalues of $H_N$ including multiplicities~\cite{julian2016}.
The harmonic approximation to the partition function for the homogeneous system can be analogously defined by
\begin{align}
    Z_{\rm h}^{\text{hom}}
    = C_{\beta,N} \cdot \big(\text{det}^{+}(H^{\text{hom}}_N)\big)^{-1/2}.
\end{align}
Hence, the harmonic approximation of the {\it Helmholtz formation free energy} is defined by
\begin{equation} \label{eq:our_def}
\begin{aligned}
\mathcal{A}(\bar{u}_N) &:= -\frac{1}{\beta} \log \left(\frac{Z_{\rm h}}{Z_{\rm h}^{\text{hom}}}\right)\\ 
&= \mathcal{E}_N(\bar{u}_N) - k_B T \Big ( -\frac{1}{2} \log \text{det}^{+} \big(H_N(\bar{u}_N)\big) + \frac{1}{2} \log \text{det}^{+} \big(H^{\text{hom}}_N\big) \Big ) \\
&= \mathcal{E}_N(\bar{u}_N) - T \mathcal{S}_N(\bar{u}_N),
\end{aligned}
\end{equation}
where the vibrational formation entropy is 
\begin{align}
\mathcal{S}_N(\bar{u}_N) := \frac{k_B}{2}\Big ( \log \text{det}^{+} \big(H^{\text{hom}}_N\big) - \log \text{det}^{+} \big(H_N(\bar{u}_N) \big)\Big ). 
\end{align}

% \co{\sout{The aforementioned framework is classically based, but it readily extends to quantum mechanics when assuming that a lattice vibration mode with a specific frequency behaves as a harmonic oscillator, leading to quantized energy levels. An overview of vibrational entropy within this quantum perspective is given in the Appendix \ref{sec:projected_normal_modes}.}} \cco{is this really useful?}

In solid state physics, it is common to employ the density of states (DOS) as the basis for defining entropy. In Appendix \ref{sec:projected_normal_modes}, we show that this perspective yields the same entropy as \eqref{eq:our_def}. 

An important material property related to the vibrational entropy is the {\em transition rate} (e.g., in the context of defect motion via diffusion). We will explore this in Section \ref{sec:sub:TST} how our techniques can also be applied in that setting.

\subsection{Site entropy}
\label{sec:sub:site_entropy}
%
% Up to this point, we have provided a formal definition of entropy and emphasized its critical role. Nevertheless, f
For systems with a large number of atoms, the computational expense associated with evaluating vibrational entropy via the eigendecomposition, or Cholesky factorisation, of the hessian matrix is prohibitive. The objective of this work is to develop a surrogate model for vibrational entropy with linear scaling cost. The key step towards that end is a spatial decomposition into local {\em site entropy} contributions. 
We adopt the spatial decomposition proposed in \cite{julian2016}, which is closely related to the one used for defining site energies in the tight-binding model~\cite{chen16}.

% \cctina{I just realized we haven't defined $\textbf{F}_N$ before this. Note to myself to add it. I think briefly citing Julian's lemma on existence of $\textbf{F}_N$ makes sense.}

To establish the aforementioned spatial decomposition, we employ a self-adjoint operator \(\textbf{F}_N : \mathcal{W}^{\text{per}}_N \rightarrow \mathcal{W}^{\text{per}}_N\) which acts as $(H^{\text{hom}}_N)^{-1/2}$. The existence and properties of $\textbf{F}_N$ has been previously established in {\cite[Lemma 2.5]{julian2016}}.

We first rewrite the entropy difference \eqref{eq:our_def} as 
% \ctina{
\begin{eqnarray} \label{eq:SN_trace}
\mathcal{S}_N (u) = -\frac{1}{2} \operatorname{Trace}\, \log^+{(\textbf{F}_N H_N (u) \textbf{F}_N)},
\end{eqnarray}
where $\log^+$ is defined as follows: 
\begin{figure}[!htb]
     \centering
     % \begin{subfigure}[b]{0.42\textwidth}
         % \centering
        \includegraphics[width=0.4\textwidth]{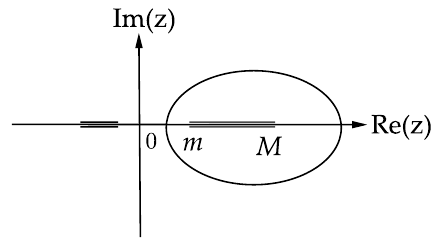}
    \caption{Depiction of the spectrum of $\mathbf{T}$, where $\sigma(\mathbf{T}) \cap (0, \infty) \subset [m,M] $ along with the appropriate contour $\mathcal{C}$.} 
    \label{fig:contour}
\end{figure}
Let $\mathbf{T}$ be a bounded, self-adjoint operator on a Hilbert space with spectrum $\sigma(\mathbf{T}) \subset (-\infty, 0] \cup [m, M]$ where $0 < m \leq M$,  as depicted in Figure \ref{fig:contour}. Then, we can define a contour $\mathcal{C}$ that encircles the interval $[m, M]$ but remains in the right half-plane. Using resolvent calculus we can now define (cf. \cite{julian2016} for more details)
\begin{equation} \label{contour3}
\log^{+} \mathbf{T}= \frac{1}{2 \pi i} \oint_{\mathcal{C}} \log (z)\cdot(z\mathbf{I}-\mathbf{T})^{-1} \,{\rm d} z.
\end{equation}
The trace operation in \eqref{eq:SN_trace} is interpreted as a spatial decomposition from~\cite{julian2016}, allowing us to write 
\begin{align}
\label{eq:SN_spatial_decomp}
\mathcal{S}_N(u) &= \sum_{\ell \in \Lambda_N} \mathcal{S}_{N,\ell}(u), \qquad  \text{where}  \\ 
\label{eq:site_entropy}
   \mathcal{S}_{N,\ell}(u) &:= -\frac{1}{2}\operatorname{Trace}
   \big[\log^+(\textbf{F}_N H_N(u) \textbf{F}_N)\big]_{\ell\ell},
\end{align}
where $[L]_{\ell\ell}$ denotes the $3 \times 3$ block of $L$ corresponding to an atomic site $\ell \in \Lambda$. The site entropy $\mathcal{S}_{N, \ell}$ will be our central object of study. 

% \cctina{Remark: mentioning these as a tiny subsection here makes a lot of sense to me.}

% \ctina{
\subsection{Thermodynamic limit}
\label{sec:sub:tlim}
While our computational investigations will be for the supercell model, the analysis is more convenient to perform in infinite lattice limit. We therefore review relevant results from \cite{julian2016, 2018-uniform, bcshap2016}.

We now consider displacement fields that are either compactly supported or of finite energy, characterized by the function spaces
\begin{align*}
\dot{\mathcal{W}}^{\rm c} &:= \left\{ u : \Lambda \rightarrow \mathbb{R}^m \mid \text{supp}(Du) \text{ is compact} \right\}, \quad \text{and} \\
\dot{\mathcal{W}}^{1,2} &:= \left\{ u : \Lambda \rightarrow \mathbb{R}^m \mid \|Du\|_{\ell^2} < \infty \right\},
\end{align*}
where 
\begin{equation}
|Du(\ell)|^2 = \sum_{\rho \in R_{\ell}} |D_{\rho}u(\ell)|^2, \qquad
\|Du\|_{\ell^2(\Lambda)} = \left(\sum_{\ell \in \Lambda} |Du(\ell)|^2 \right)^{1/2}.
\end{equation}
The above expressions define a semi-norm for both $\dot{\mathcal{W}}^c$ and $\dot{\mathcal{W}}^{1,2}$ spaces.
The energy functionals for the homogeneous and defective lattice are defined respectively as 
\begin{align}
\E_{\text{hom}}(u) &= \sum_{\ell \in \Lambda} V\big(Du(\ell)\big), \\
\E(u) &= \sum_{\ell \in \Lambda} V_{\ell}\big(Du(\ell)\big) \quad \text{for } u \in \dot{\mathcal{W}}^c.
\end{align}
We now consider the equilibrium configurations, $\delta \E(\bar{u}) = 0$ which can equivalently be written as
\begin{equation} \label{eq:variational_inf}
    \langle \delta \E(\bar{u}) , v\rangle = 0.
\end{equation}

\begin{theorem}[Thermodynamic Limit]{\cite[Theorem 2.1]{2018-uniform}} 
    \label{thm:thermo_limit}
    Let $\bar{u} \in \dot{\mathcal{W}}^{1,2}$ be a stable solution to \eqref{eq:variational_inf}. For $N$ sufficiently large, there exists a locally unique solution to \eqref{eq:opt} such that the following estimates hold:
    \begin{align}
        \|D\bar{u}_N - D\bar{u}\|_{\ell^{\infty}(\Lambda_N)} &\leq CN^{-d}, \\
        \|D\bar{u}_N - D\bar{u}\|_{\ell^{2}(\Lambda_N)} &\leq N^{-d/2}, \\
        |\mathcal{E}_N(\bar{u}_N) - \mathcal{E}(\bar{u})| &\leq N^{-d}.
    \end{align}
    Furthermore,
    the thermodynamic limit of $\mathcal{S}_N$ as $N \to \infty$, denoted by $\mathcal{S}$, exists. The error in approximating $\mathcal{S}(\bar{u})$ with $\mathcal{S}_N(\bar{u}_N)$ satisfies {\cite[Theorem 2.3]{julian2016}} 
    \begin{equation} \label{julian_err_S}
        \left|\mathcal{S}(\bar{u}) - \mathcal{S}_N(\bar{u}_N)\right| \lesssim N^{-d} \log^5(N).
    \end{equation}
\end{theorem}

% \begin{theorem}{\cite[Theorem 2.3]{julian2016}}     
%     \label{eq:minimizer}
%     \cco{Isn't this result from \cite{2018-uniform}}
%     Let $\bar{u} \in \dot{\mathcal{W}}^{1,2}$ be a stable solution to \eqref{eq:variational_inf}. For $N$ sufficiently large, there exists a locally unique solution to \eqref{eq:opt} such that the following estimates hold:
%     \begin{align}
%         \|D\bar{u}_N - D\bar{u}\|_{\ell^{\infty}(\Lambda_N)} &\leq CN^{-d}, \\
%         \|D\bar{u}_N - D\bar{u}\|_{\ell^{2}(\Lambda_N)} &\leq N^{-d/2}, \\
%         |\mathcal{E}_N(\bar{u}_N) - \mathcal{E}(\bar{u})| &\leq N^{-d}.
%     \end{align}
% \end{theorem}

% It is then shown in \cite{julian2016} that the thermodynamic limit of $\mathcal{S}^+_N$ as $N \to \infty$,  denoted by $\mathcal{S}^+$, exists and the error satisfies 
% \begin{eqnarray}\label{eq:limit}
% \left|\mathcal{S}^+(\bar{u}) - \mathcal{S}^+_N(\bar{u}_N)\right| \lesssim N^{-d} \log^5(N).
% \end{eqnarray}

Additionally, in~\cite[Lemma 3.2]{julian2016}, it is shown that there exist $\underline{\sigma}, \bar{\sigma}, \varepsilon >0 $ and a contour $\mathcal{C}$ encircling $[\underline{\sigma}, \bar{\sigma}]$ but not the origin, such that, if some $u_\infty \in \dot{\mathcal{W}}^{1,2}$ satisfies $\sigma \big ({\rm \bF}^{*} H(u_\infty)\bF \big) \cap (-\underline\sigma, \infty) \subset [2\underline\sigma, \overline\sigma/2]$ for some $0 < \underline\sigma < \overline\sigma$, then for all $u\in B_{\varepsilon}(u_{\infty})$, the operator $\pmb{\rm F}^{*}  H(u) \pmb{\rm F} $ remains uniformly bounded above and below and 
\begin{align}
\log^{+}\big[{\rm \bF}^{*}H(u)\bF\big] =&~\frac{1}{2 \pi i} \oint_{\mathcal{C}} \log (z)\cdot \mathscr{R}_z(u)\,{\rm d}z, \nonumber \\
\textrm{where} ~\mathscr{R}_z = \mathscr{R}_z(u) :=&~(zI- \pmb{\rm F}^{*}  H(u) \pmb{\rm F})^{-1}.
    \label{eq:resolvent}
\end{align}
From now on, we will fix this contour and always have $z\in\mathcal{C}$. We will also express the limit quantity for each site entropy $S_{N,\ell}$ using the generalized notation
\begin{eqnarray}\label{eq:local_S+}
\mathcal{S}_{\ell}(u) := -\frac{1}{2}\operatorname{Trace}\big[\log^+(\textbf{F}^*H(u)\textbf{F})\big]_{\ell\ell}\quad\textrm{with}\quad 
\textbf{F} := (H^{\text{hom}})^{-1/2} \in \mathcal{L}(\ell^2, \dot{\mathcal{W}}^{1,2})
\end{eqnarray}
where $(H^{\text{hom}})^{-1/2}$ denotes the square root of the pseudo-inverse of $H^{\text{hom}}$. For a more detailed discussion on the existence of $\textbf{F}$ and a rigorous definition using Fourier transform see Appendix~\ref{sec: appendix_prelem}. 
% Additionally, since $\ell^{2}$ excludes constant displacements, the formulation does not necessitate a projection operator similar to \( \pi_N \).

% Since $\ell^2(\Lambda)$ does not contain any constant displacements, there is no need for a projector analogous to $\pi_N$ in the definition of $\textbf{F}$.

The homogeneous site entropy can be similarly defined by
\begin{equation}
\mathcal{S}^{\text{hom}}_\ell(u) := -\frac{1}{2}\operatorname{Trace}\big[\log^+(\textbf{F}^*H^{\text{hom}}(u)\textbf{F})\big]_{\ell\ell} .
\end{equation}
% \cys{Is the following discussion really needed?} The operator $\log(\textbf{F}^*H(u)\textbf{F})$ cannot be expected to be of trace class. Thus, $\mathcal{S}(u)$ cannot be simply defined as $\mathcal{S}(u) := -\frac{1}{2}\operatorname{Trace}\log(\textbf{F}^*H(u)\textbf{F})$ which would be the sum of the site contributions $\mathcal{S}_\ell(u)$, and a more precise definition of $\mathcal{S}(u)$ is required. For details of this definition and a rigorous definition of $\textbf{F}$ via Fourier transform see \cite{julian2016}. 
% For a derivation of site entropies using density of states see
% \eqref{sec:apd:S-DOS}.
Following the arguments in the proof of \cite[Proposition 5.8]{julian2016}, the existence of a similar limit quantity for $\mathcal{S}_{N, \ell}(\bar{u}_N)$, namely $\mathcal{S}_\ell(\bar{u})$ can be deduced. Thus, instead of analyzing the supercell approximation of site entropy $\mathcal{S}_{N, \ell}(u)$, we will give a locality estimate of $\mathcal{S}_{\ell}(u)$ (cf.~Theorem~\ref{thm:local}) in the subsequent section.

% \ctina{A typical approach in Solid State Physics which is discussed in Appendix \ref{sec:projected_normal_modes} is to determine the site vibrational entropy using the projected normal modes on each atom as follows:
% \begin{align}
% \mathcal{S}_{\ell}(\omega) &:= \langle \Omega_\ell(\omega), \xi \rangle = \sum_{\alpha =1}^{3N}  \xi(\omega_\alpha)  [\phi_\alpha]_{\ell}^2 \quad \text{where} \  \xi := k_B \left[\log\left(\frac{kT}{\hbar\omega_\alpha}\right) + 1\right],
% \end{align}
% where \(|\phi_\alpha \rangle\) and \(\omega_\alpha^2\) represent the eigenstates and eigenvalues of the Hessian \(H\), \([\phi_\alpha]_{\ell}\) is the \(\ell\)-th entry of \(\phi_\alpha\) and \(\Omega_\ell\) is the projected density of states defined as follows:
% \begin{align} \label{eq:decomDOS}
%     \Omega_\ell(\omega)= \sum_{\alpha} \delta(\omega - \omega_\alpha) [\phi_\alpha]_{\ell}^2,
% \end{align}
% It is straightforward to verify that
% \begin{align}
%     \mathcal{S}_{\text{total}}&= \sum_{\ell\in\Lambda} \mathcal{S}_\ell = k_B \sum_{\alpha =1}^{3N}  \left[\log\left(\frac{kT}{\hbar\omega_\alpha}\right) + 1\right]  [\phi_\alpha]_{\ell}^2. 
% \end{align}
% In the local basis, total entropy can be accurately partitioned into site-specific contributions. This local entropy is essential for understanding the thermodynamic properties of materials at the atomic level. Thus, the locality principle is not just a mathematical convenience but a reflection of the physical reality of how atomic vibrations contribute to the overall entropy of a system.
% }

\section{Locality of Site Entropy}
\label{sec:th_rs}
We begin by thoroughly analyzing the locality of site entropy, establishing a foundation for its subsequent impact on the error analysis arising from the truncation of the interaction range in Section~\ref{sec:sub:local}. To substantiate our theoretical results, a numerical validation of the locality estimate is presented in Section~\ref{sec:sub:numer}.

\subsection{Locality Estimates for $\mathcal{S}_\ell$}
\label{sec:sub:local}

It is shown in~\cite{julian2016} that each individual entropy $\mathcal{S}_\ell$ has only a small dependence on distant atomic sites. A concrete representation of this decay is provided by the formal estimate
\begin{equation}
\left| \frac{\partial \mathcal{S}_\ell(\bar{u})}{\partial Du_n} - \frac{\partial \mathcal{S}^{\text{hom}}_\ell({\bf 0})}{\partial Du_n} \right| \lesssim |\ell - n|^{-2d}|n|^{-d} + \text{higher order terms}.
\label{eq:julian_locality}
\end{equation}
This equation provides an initial glimpse into the localized characteristics of $\mathcal{S}_{\ell}$ and sheds light on why one might anticipate effective control over its summation across $\ell$. A more precise estimate for this locality will be established in the remainder of the section, building on the machinery developed in \cite{julian2016}. Our goal in this section is to show that site entropies have a locality property.

% It is a fundamental result that for an analytic function $f$, applied to a square matrix $\mathbf{T}$, one can express this as a contour integral in the complex plane. Specifically, the matrix function $f(\mathbf{T})$ is given by
% \begin{equation}
% f(\mathbf{T}) = \frac{1}{2\pi i} \oint_{\mathcal{C}} f(z)(z\mathbf{I} - \mathbf{T})^{-1} \, dz,
% \end{equation}
% \cco{this motivation paragraph won't make sense to anybody who hasn't read Julian's or Huajie's papers. I think it's better to move this into a "Sketch of the proof" after the statement of the theorem. That will work quite well I think.}
% where $\mathcal{C}$ denotes a closed contour that resides within the domain where $f$ is analytic, encircling the spectrum $\sigma(\mathbf{T})$ exactly once in a counterclockwise direction. We will use a generalization which allows us to extend our analysis to operators whose spectra are not entirely confined to the positive half-plane. This generalization is elaborated in Appendix \ref{sec: appendix_prrof}. Using this generalization  will define $\mathcal{S}^+_\ell$ as follows:
% \begin{equation} \label{eq:defSell2}
% \mathcal{S}^+_\ell(u) := -\frac{1}{2}{\rm Trace} \Big[\log^+ \big(\pmb{\rm F}^{*}  H(u) \pmb{\rm F}\big)\Big]_{\ell\ell},
% \end{equation}
% \cco{Shouldn't this be exactly the same of \eqref{eq:local_S}? If it is not then this needs to be explained. If it is the same then why introduce a new term? I think it would actually be better to introduce $\log^+$ in the previous section.}
%  Our goal in this section is to show that site entropies have a locality property. 

\begin{theorem}\label{thm:local}{\rm\textbf{(Locality)}}
Given $u \in \dot{\mathcal{W}}^{1,2}$, let the site entropy $\mathcal{S}_{\ell}(u)$ be defined by~\eqref{eq:local_S+}.  Then there exist a constant $C_2>0$ such that, for $\ell, n\in \Lambda$, and for $r_{n\ell}$ sufficiently large, 
\begin{align}
    \left| \frac{\partial \mathcal{S}_\ell(u) }{\partial u_n} \right|  \leq C_2 \big|r_{n\ell}\big|^{-2d},% \cdot \log(1+ r_{n\ell}),
\end{align}
where $r_{n \ell}$ represents the distance between atoms $n$ and $\ell$. 
\end{theorem} \\
% \ctina{\begin{sketchoftheproof}
% We begin by expressing the site entropy $\mathcal{S}^+_{\ell}$ using a fundamental result that for an analytic function $f$, applied to a square matrix $\mathbf{T}$, one can express this as a contour integral in the complex plane. Specifically, the matrix function $f(\mathbf{T})$ is given by
% \begin{equation}
% f(\mathbf{T}) = \frac{1}{2\pi i} \oint_{\mathcal{C}} f(z)(z\mathbf{I} - \mathbf{T})^{-1} \, dz,
% \end{equation}
% where $\mathcal{C}$ denotes a closed contour that resides within the domain where $f$ is analytic, encircling the spectrum $\sigma(\mathbf{T})$ exactly once in a counterclockwise direction.
% By expressing $\log^+ \big(\pmb{\rm F}^{*}  H(u) \pmb{\rm F}\big) $ using a contour integral and differentiating it with respect to $u$, we then utilize Lemma~\ref{th:properties_F} and resolvent estimates from Lemma~\ref{th:resolvent estimate}, to obtain an estimate that scales with the distance between the atoms. The complete proof is given in Appendix~\ref{sec: appendix_prrof}.
% \end{sketchoftheproof}}

\medskip 

Theorem \ref{thm:local} reveals the locality principle of site entropy, which indicates that the sensitivity of site entropy to displacements decreases algebraically as the distance $r_{n\ell}$ grows. We will numerically validate this theorem in Section~\ref{sec:sub:numer}. Owing to this locality, the accurate approximations of $\mathcal{S}^+_{\ell}(u)$ can be achieved by limiting the group of atoms with index $j$ to a certain vicinity around atom $\ell$, i.e., $r_{\ell j} < r_{\text{cut}}$. More precisely, given the site $\ell \in \Lambda$, we define the $\textit{truncated site entropy}$ $\widetilde{\mathcal{S}}^+_\ell(u)$ by 
% \begin{equation}
% \widetilde{\mathcal{S}}^+_{\ell}(u; r_{\rm cut}) := \mathcal{S}^+_{\ell}(\hat{u}; r_{\rm cut}) \quad \text{with} \quad 
% \hat{u}_{k} := 
% \begin{cases} 
% u_{k} & \text{for } k \in \Lambda \cap B_{r_{\rm cut}}(\ell) \\
% \pmb{0} & \text{otherwise} 
% \end{cases}. \label{eq:truncated_def}
% \end{equation}

\begin{equation}
\tilde{\mathcal{S}}_\ell(u) = \mathcal{S}_\ell( T_{r_{\rm cut}} u) \quad \text{where} \quad 
T_{r_{\rm cut}} u(k) := 
\begin{cases} 
u_k & \text{for } k \in \Lambda \cap B_{r_{\rm cut}}(\ell) \\
\pmb{0} & \text{otherwise}
\end{cases}. \label{eq:truncated_def}
\end{equation}

The subsequent theorem suggests that the truncated site entropy $\widetilde{\mathcal{S}}_{\ell}$ defined in \eqref{eq:truncated_def}, serves as a reliable approximation to the non-truncated entropy $\mathcal{S}_{\ell}$. We give the complete proof of this theorem in the Appendix \ref{sec: appendix_prrof} and will be numerically verified in Section \ref{LR}.

\begin{theorem}\label{thm:err_local}
Let $\widetilde{\mathcal{S}}_\ell({u};r_{\rm cut})$ defined by \eqref{eq:truncated_def} for $u \in \dot{\mathcal{W}}^{1,2}$, denote the truncated site entropy for atomic site $\ell$. Then there exists a constant $C_3>0$ such that for $r_{\rm cut}$ sufficiently large, 
\begin{align}
    \big| \mathcal{S}_\ell(u) - \widetilde{\mathcal{S}}_\ell(u ; r_{\rm cut}) \big| \leq
     C_3| r_{\rm cut} |^{-d} \cdot \| u \|_{L^{\infty}}. 
\end{align}
\end{theorem}

\medskip 

The two foregoing theorems justify calling \eqref{eq:SN_spatial_decomp} a {\em spatial decomposition}. It further motivates developing a surrogate model for vibrational entropy in terms of a sum of parameterized site entropy contributions. We will explore this further in Section~\ref{sec:num_rs}. In the recent study~\cite{PhysRevMaterials.4.063802} the spatial localization of site entropy was postulated as an {\em assumption} to construct surrogate models for entropy.
% assumption of  was used to construct surrogate models for entropy. However, the proof of the locality property, as well as its examination through numerical tests, was not undertaken in their work.}

% , make the same as is done in the machine learning of interatomic potentials~\cite{ PhysRevB.87.184115, PhysRevLett.104.136403}. 
% we can treat site entropy the same as site energy %\cys{say one sentence that because of this theorem, we can treat the site entropy exactly the same as the site-energy}% 
% and as a result, construct models that correlate local physical observables with atomic-scale descriptors . 
% This opens the pathways of fitting entropy using machine-learning techniques in the same way as interatomic potentials (cf.~). 

\subsection{Numerical validation}
\label{sec:sub:numer}
To validate the locality of site entropy (cf.~Theorem~\ref{thm:local}), we consider a toy model with pairwise interactions \cite{GADES1992232, ZIEGENHAIN20091514}. 
We consider lattice displacements as functions $u: \Lambda \to \mathbb{R}$ and define the energy as:
\begin{equation}\label{eq:toymodel}
E(C; u) = \sum_{i} \left( \sum_{j \in \mathcal{N}_i} C_{ij} |u_i - u_j|^2 + \delta |u_i - u_j|^3 + |u_i - u_j|^4 \right),
\end{equation}
where \( i \) indexes the lattice points, \( \mathcal{N}_i \) denotes the set of nearest neighbors of the \( i \)-th site, \( u_i \) represents the displacement at the \( i \)-th lattice site, \( C_{ij} \) are the coefficients characterizing the interactions between points \( i \) and \( j \), and \( \delta \) is a constant which controls the contribution of the strength of the cubic term. (Qualitatively, this is not a severe restriction of generality, see e.g. \cite[Section 6.2]{julian2016}).

% \cco{arxiv:2108.04765, Sec. 6.2; please track the publication and cite it}.)
 This form encapsulates the 1D, 2D cases with appropriate definitions of \( \mathcal{N}_i \) for each dimension. When defining the matrices $C$, we only consider nearest neighbour interactions and employed periodic boundary conditions. Bonds are indicated by the interaction coefficient values, with 1 signifying a bond and 0 its absence. To incorporate inhomogenity into our numerical tests, we integrated ``impurities'' into the model by introducing perturbations to the existing elements of the matrix $C$. The inclusion of defects and impurities in the system does not impact the locality property of site entropy.
% we consider 1,2, and 3 dimesional energy toy models including $N$ atoms. 

It is noteworthy that the employment of the toy model here instead of real systems, enables us to easily perform large-scale simulation in which we can most clearly observe the anticipated locality behavior. 

% \co{\sout{ We are 
% confident that testing on real systems makes no qualitative difference.}} \cco{I don't like saying that without evidence.} 
% The energy expressions in one, two, and three dimensions can be written in a unified general form by considering the displacement between neighboring lattice points. The general form of the energy \( E \) in a \( d \)-dimensional lattice is given by:

 % We will numerically study and present the result on its effects on the locality in this section. 

\begin{figure}[!htb]
     \centering
          \begin{subfigure}[b]{0.32\textwidth}
         \centering
         \qquad \qquad \qquad
         \includegraphics[width=\textwidth]{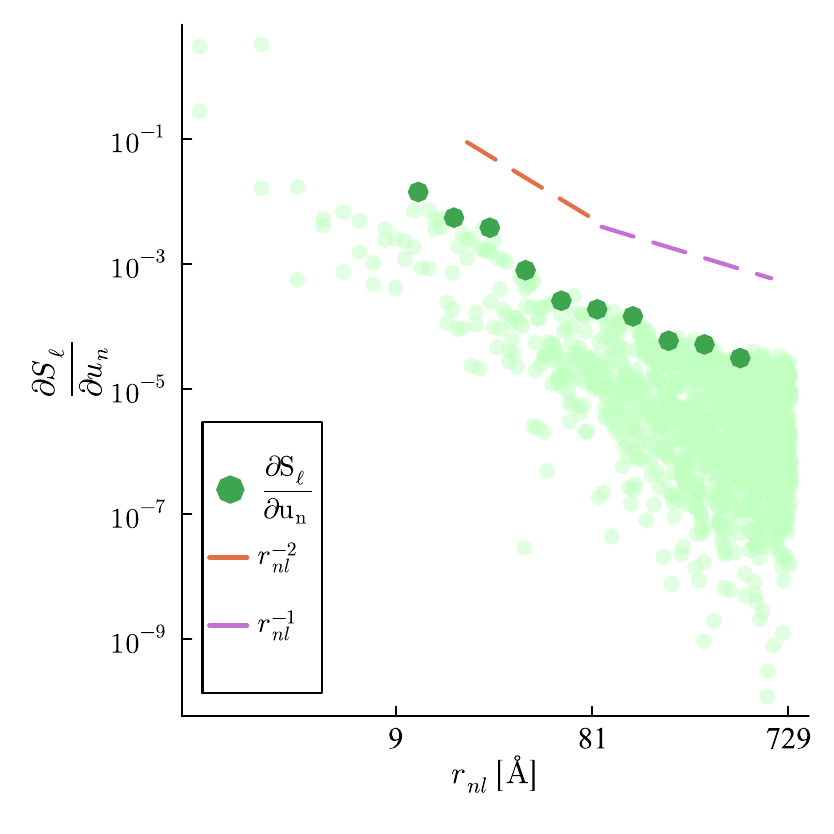}
         \caption{$\delta = 0$}
         \label{fig:delta0.0_1D}
     \end{subfigure}
     \begin{subfigure}[b]{0.32\textwidth}
         \centering
        \includegraphics[width=\textwidth]{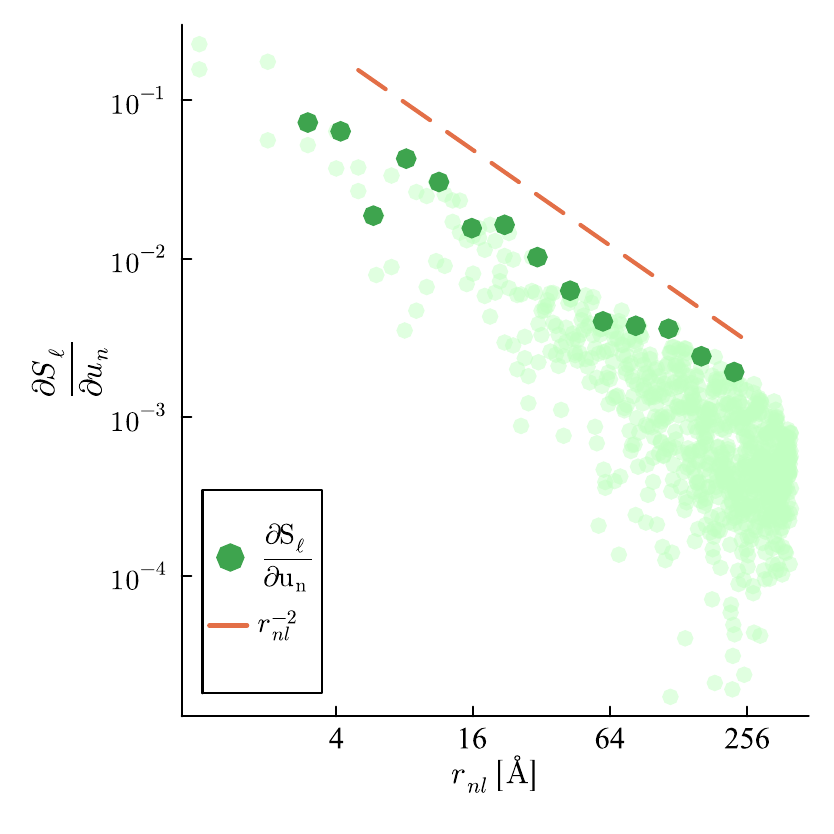}
         \caption{$\delta = 0.1$}
         \label{fig:delta0.1_1D}
     \end{subfigure}
     \begin{subfigure}[b]{0.32\textwidth}
         \centering
         \qquad \qquad \qquad
         \includegraphics[width=\textwidth]{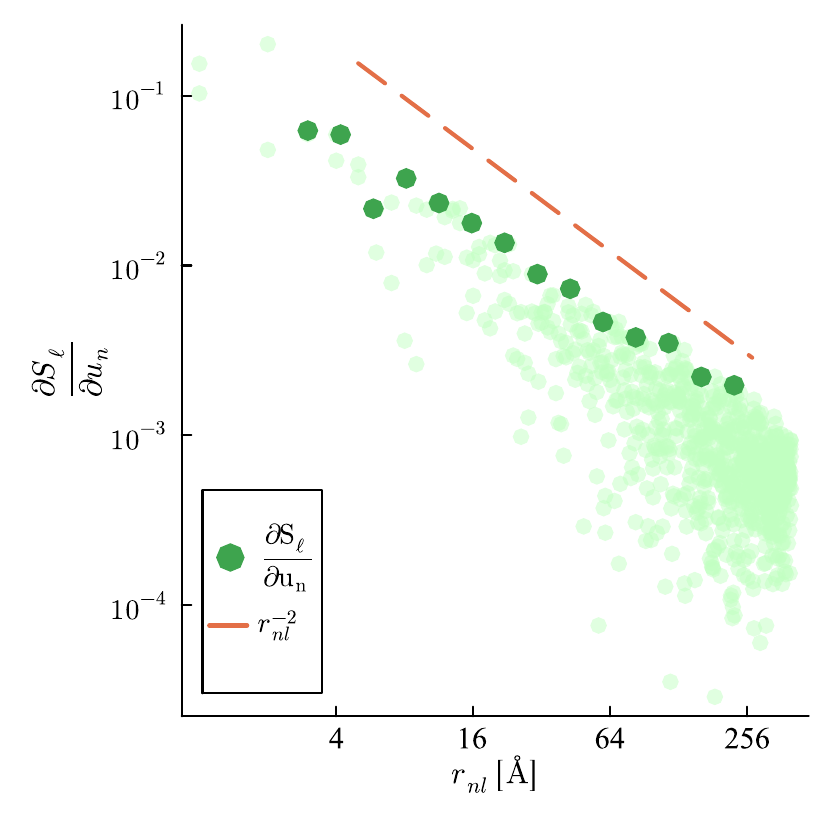}
         \caption{$\delta = 0.5$}
         \label{fig:delta0.5_1D}
     \end{subfigure}
     
    \caption{Derivatives of site entropy for the 1 dimensional toy model including 800 atoms. The dark green dots represent the maximum $\frac{\partial S}{\partial u_n}$ values within specified logarithmic bins along the $x$ axis.}
    \label{fig:1dtoy}
\end{figure}

\begin{figure}[!htb]
     \centering
     \begin{subfigure}[b]{0.32\textwidth}
         \centering
    \includegraphics[width=\textwidth]{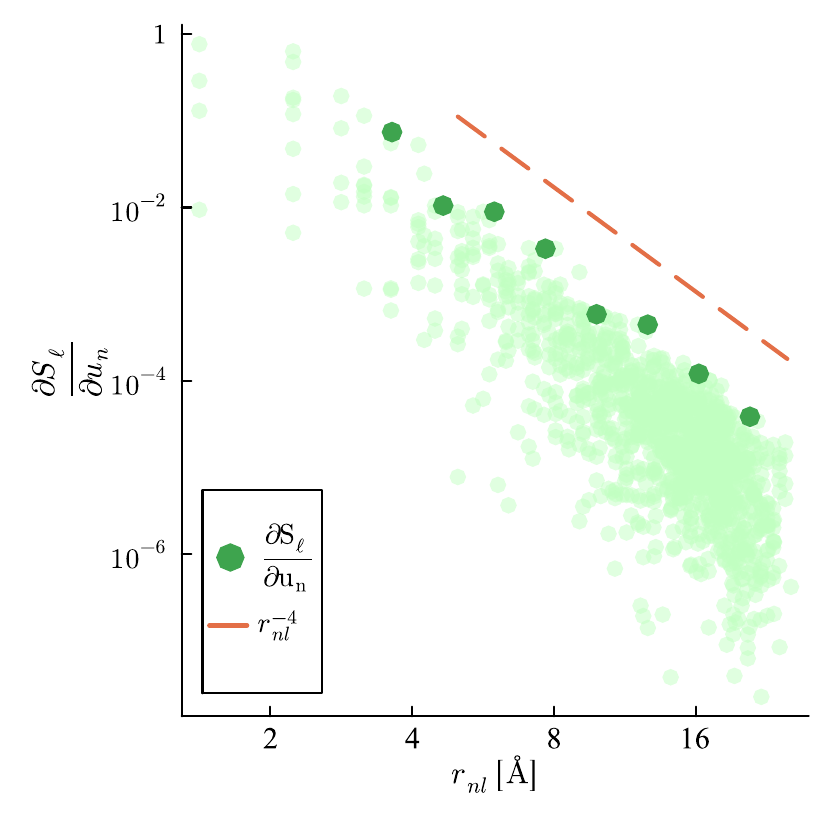}
         \caption{$\delta = 0$}
         \label{fig:delta0_2D}
     \end{subfigure}
     \begin{subfigure}[b]{0.32\textwidth}
         \centering
    \includegraphics[width=\textwidth]{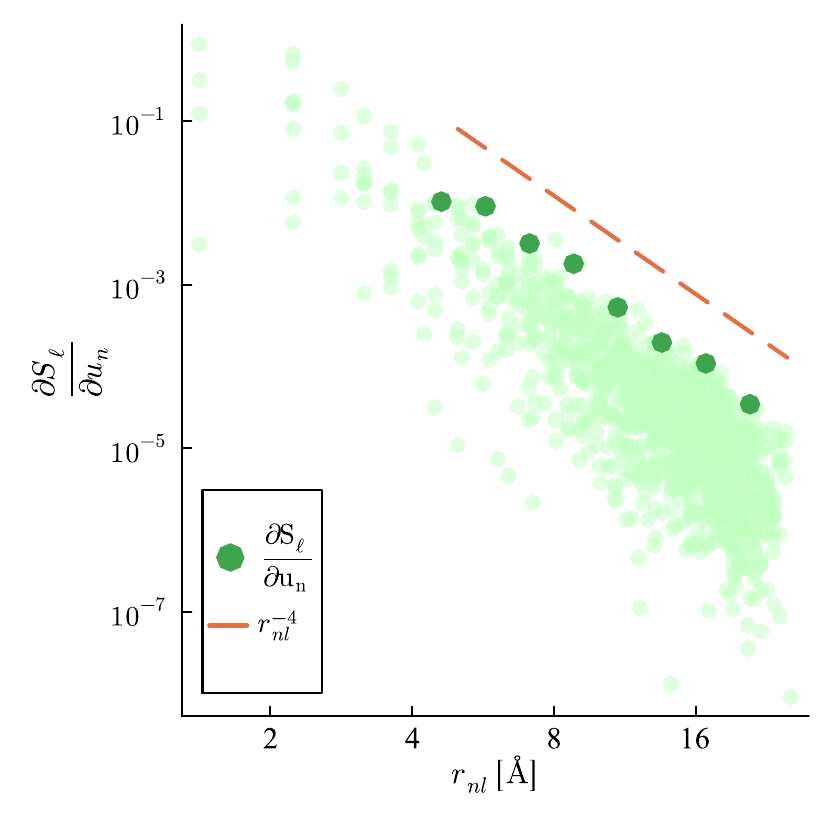}
         \caption{$\delta = 0.1$}
         \label{fig:delta0.1_2D}
     \end{subfigure}
     \begin{subfigure}[b]{0.32\textwidth}
         \centering
         \qquad \qquad \qquad
         \includegraphics[width=\textwidth]{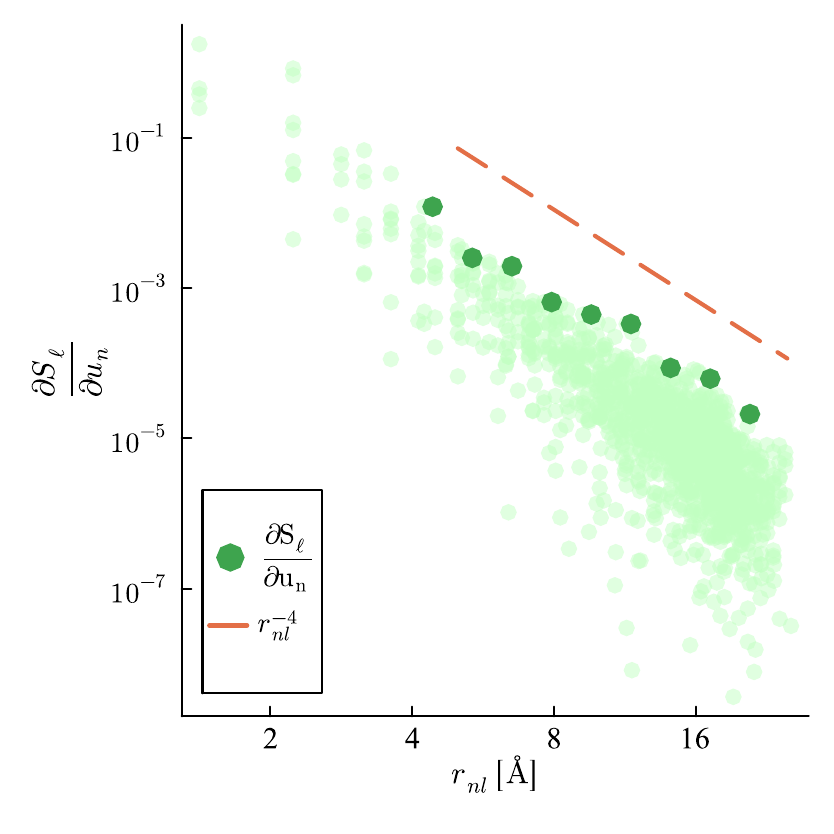} 
         \caption{$\delta = 0.5$}
         \label{fig:delta0.5_2D}
     \end{subfigure}
    \caption{Derivatives of site entropy for the 2 dimensional toy model including 2025 atoms. The dark green dots represent the maximum $\frac{\partial S}{\partial u_n}$  values within specified logarithmic bins along the x axis.}
    \label{fig:2dtoy}
\end{figure}

Following the detailed representations of our models, we conducted numerical evaluations of the derivatives of the site entropy, $ \left| \frac{\partial S_\ell}{\partial u_n} \right| $.
Figures \ref{fig:1dtoy} and \ref{fig:2dtoy}, demonstrate the rate of decay of site entropies.
Our numerical results for the $2D$ model are in good agreement with the predictions made by Theorem~\ref{thm:local}, irrespective of the $\delta$ value. In particular, as we observe larger values of \( |r_{n\ell}| \), the magnitude of \( \left| \frac{\partial \mathcal{S}_\ell(u) }{\partial u_n} \right| \) decreases, and this decrease follows the \( |r_{n\ell}|^{-4} \) rate indicated by our theory. In the context of the $1D$ model, the numerical outcomes align with theoretical expectations when $\delta = 0.1$ and $\delta = 0.5$. However, for $\delta = 0$, we observed a pre-asymptotic rate of $|r_{n\ell}^{-2}|$, which soon transitions to $|r_{n\ell}^{-1}|$. An enlargement of the domain size revealed that the true rate for this scenario is in fact $|r_{n\ell}^{-1}|$. This can potentially be attributed to the presence of symmetries and cancellation effects in the energy model when the cubic term is disregarded.
In summary, our numerical results strongly support the qualitative sharpness of our theory.
% and numerical results affirms the reliability of our theory and underscores the importance of the locality of site entropies.
 
 % It is also noteworthy that the parameter $\delta$ introduced into the model in~\eqref{eq:toymodel}, does not effect the predicted decay rate, which provides numerical evidence for the sharpness of our estimate.

\begin{figure}[!htb]
     \centering
     \begin{subfigure}[b]{0.42\textwidth}
         \centering
    \includegraphics[width=\textwidth]{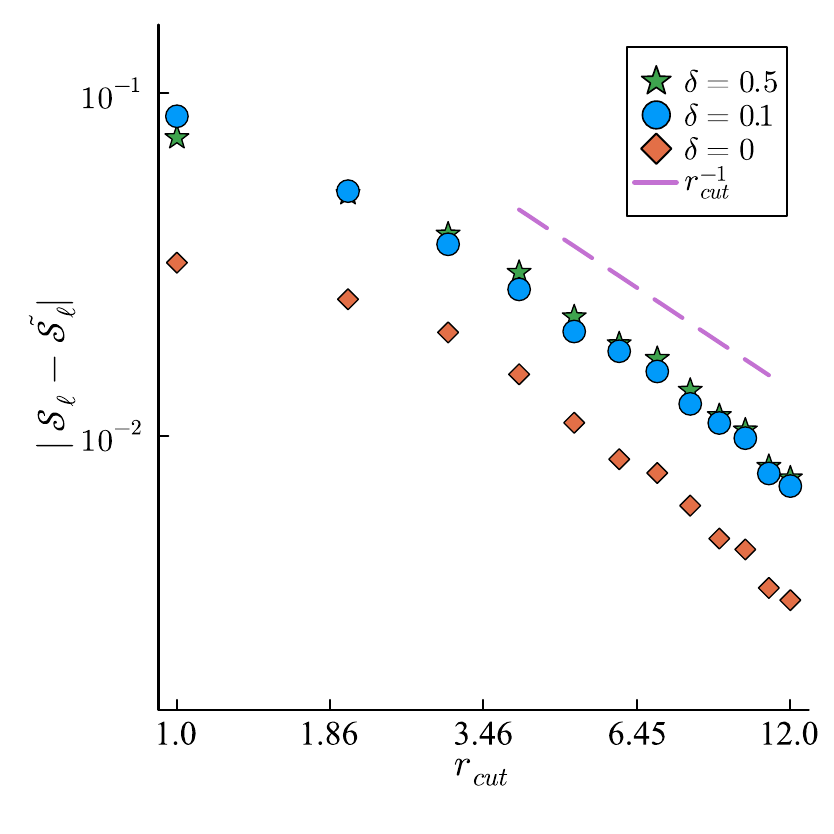}
         \caption{$1$D model}
         \label{fig:rc_1D}
     \end{subfigure}
     \begin{subfigure}[b]{0.42\textwidth}
         \centering
    \includegraphics[width=\textwidth]{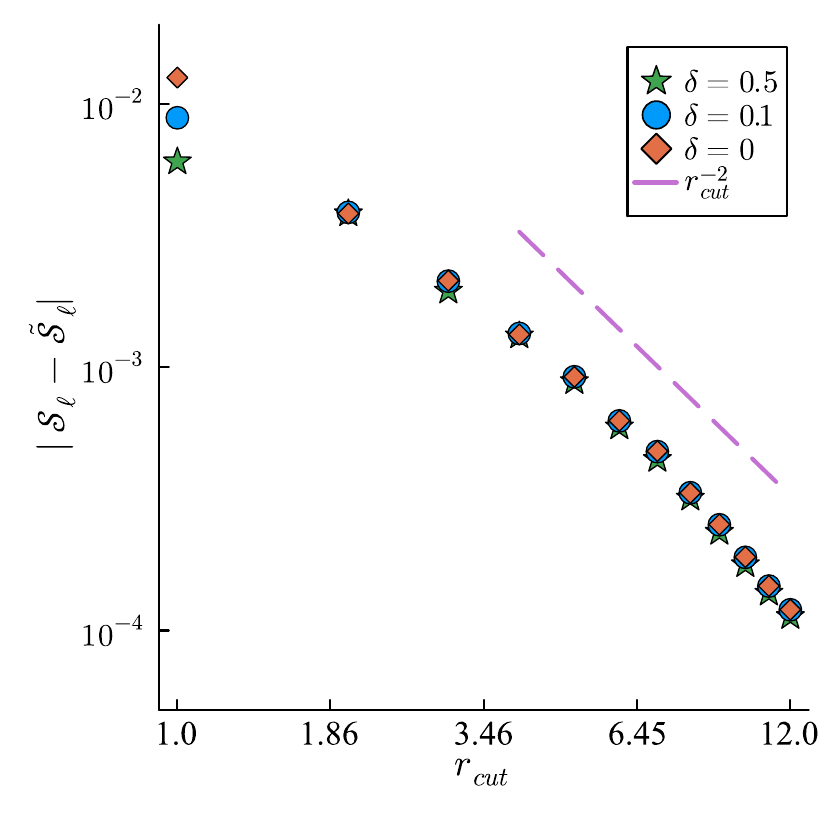}
         \caption{$2$D model}
         \label{fig:rc_2D}
     \end{subfigure}
    \caption{Absolute difference in truncated site entropy $\widetilde{\mathcal{S}}_{\ell}$ and non-truncated site entropy $\mathcal{S}_{\ell}$ for the 1D (\textbf{left}) and 2D (\textbf{right}) toy models with different $\delta$ values. The 1D model included 800 atoms, while the 2D model included 4225 atoms, with the site entropies evaluated near the center of the domains.}
    \label{fig:toy_rc}
\end{figure}

Figure \ref{fig:toy_rc} illustrates the relationship between the difference in truncated site entropy, denoted as $\widetilde{\mathcal{S}}_{\ell}$, and the non-truncated site entropy, $\mathcal{S}_{\ell}$, as a function of the cutoff radius. This relationship is shown separately for one-dimensional (1D) and two-dimensional (2D) toy models. The observed trend indicates that the reduction in the difference between $\widetilde{\mathcal{S}}_{\ell}$ and $\mathcal{S}_{\ell}$ follows the decay rate $|r_{\rm cut}|^{-d}$, as proposed in theorem \ref{thm:err_local}. This trend is consistent across both the 1D and 2D models. This numerically confirms that the truncated site entropy serves as an effective approximation for the non-truncated site entropy and provides a second justification for the surrogate models we will propose in the next section.

% \cco{I think there should also be a test for the second locality result.}

It should be noted that the evaluation of site entropy derivatives is computationally demanding.  To address this, we adopted an approach to identify the optimal contour and used the sparse nature of the Hessian Jacobians for derivative calculation. The details of this methodology are elaborated in Appendix \ref{sec:apd:contour}. Additionally, all source codes for the numerical tests can be found in our repository \cite{gitACEntropy}.

\section{Surrogate Models} \label{sec:num_rs}

The computation of vibrational entropy, as outlined in \eqref{eq:our_def} within the harmonic approximation, involves an initial step of determining the Hessian matrix, a difficult computation, to perform that typically scales as $O(N^2)$ if performed with finite differences or automatic differentiation, and challenging to implement ``by hand''.
Subsequently, diagonalizing this matrix incurs a scaling of $O(N^3)$ for systems comprising $N$ atoms. Exploiting sparsity can only partially reduce this cost.
% For instance, evaluating the vibrational entropy of a crystal with $10^5$ atoms necessitates over 20 terabytes of memory and a computational effort of 10 hours on a cluster of contemporary CPUs, as reported by Lapointe et al.~\cite{PhysRevMaterials.4.063802}.
To address this computational challenge, Lapointe et al.~\cite{PhysRevMaterials.4.063802} proposed a surrogate model for assessing harmonic vibrational entropy using a linear-in-descriptor machine learning (LDML) approach. 
Their derivation of entropy closely aligns with the methodology presented in Appendix \ref{sec:projected_normal_modes}, and is formally equivalent to our own derivation. The key justification for either approach is our proof of locality of site entropy.

% their work lacks a demonstration of the crucial property of locality in site entropy. This property is pivotal as it enables the fitting of entropy using descriptors of atomic environments.

In this section we connect our locality results to learning a surrogate for the site entropy functional. Specifically, we employ the Atomic Cluster Expansion (ACE) framework, and present a series of numerical results pertaining to the fitting of entropy and its applications in materials science, but with a focus on demonstrating the validity of our analysis in the previous sections.

% The LDML method predicts vibrational entropy \( S \) using a linear regression on descriptor functions of atomic coordinates. For each atom \( i \), a descriptor vector \( \uline{D}_i \) is computed from its local atomic environment. The local entropy \( s_i \) is then given by $s_i = \uline{w}_i \cdot \uline{D}_i$ with \( \uline{w} \) being the model weights. The total entropy \( S \) is the sum of local entropies:
% \[
% S = \sum_{i=1}^{N} s_i = \uline{w} \cdot \sum_{i=1}^{N} \uline{D}_i = N \uline{w} \cdot \langle \uline{D} \rangle,
% \]
% where \( \langle \uline{D} \rangle \) is the averaged descriptor vector. The entropy of formation \( S_f \) for defects is modeled as:
% \[
% S_f = (N_b \pm N_d) \uline{w} \cdot [\langle \uline{D} \rangle_d - \langle \uline{D} \rangle_b],
% \]
% indicating the mapping from the high-dimensional input space \( \mathbb{R}^{3N} \) to a reduced descriptor space \( \mathbb{R}^D \times N \).

\subsection{Atomic Cluster Expansion (ACE)}
\label{sec:sub:ACE}

% We propose using the Atomic Cluster Expansion (ACE) framework for fitting entropy. The following construction follows directly from the formulation of interatomic potentials using ACE \cite{2019-ship1, PhysRevB.99.014104}.

% The most accurate energy and force calculations are obtained from the electronic structure techniques, which are based on the direct quantum-mechanical treatment of the electrons. DFT calculations have a solid physical basis and are highly accurate. They can be used to elemental and multi-component systems with nearly equal computational effort. Therefore they are a very useful tool for a thorough investigation of materials chemistry. Meanwhile, many material processes have intrinsic length and time scales that are significantly larger than those that can currently be reached by DFT calculations. Plastic deformation, fracture, phase nucleation and growth (such as alloy melting and solidification), and the evolution of the microstructure through solid-solid interface migration are a few examples \cite{Mishin2021}. Thus, the extensive computational demands of DFT limit its use for larger systems and longer timescales. This has led to the adoption of interatomic potentials and more recently, machine-learned interatomic potentials (MLIPs) that are trained on DFT data \cite{2019-ship1, PhysRevLett.104.136403, PhysRevLett.98.146401, Shapeev2016MomentTensor}. These MLIPs are now a staple in computational materials science due to their high fidelity \cite{doi:10.1021/jacs.3c04030, Rosenbrock2021}.

Accurate energy and force calculations are best achieved through electronic structure techniques like DFT. However, DFT has limitations for larger systems and longer timescales. To overcome these limitations, machine-learned interatomic potentials (MLIPs), trained on DFT data, have become essential in computational materials science due to their high accuracy and transferrability~\cite{2019-ship1, PhysRevLett.104.136403, PhysRevLett.98.146401, 
doi:10.1021/jacs.3c04030, wang2024theoretical, Shapeev2016MomentTensor}. This methodology can be directly transferred to learn surrogate models for entropy.

We utilize the linear atomic cluster expansion (ACE) parameterisation~\cite{PhysRevMaterials.6.013804, wang2024theoretical}. All in-use MLIPs for materials typically express the total energy, $\E_{\text{ML}}$, as an sum over individual site energies. Leveraging the locality of site entropy, as established in Theorem~\ref{thm:local}, the same approach employed for energy modeling can be seamlessly applied.
% Since ACE has been thoroughly detailed elsewhere, we only summarize the key aspects: 

To that end, we consider $N$ atoms described by their position vectors $y_j$. A set $\mathbf{R} := \{y_1, \ldots, y_N\} \in \mathbb{R}^{3N}$ of $N$ particle positions is called an atomic configuration. 
Let $y_{ij} = y_j - y_i$ be the distances between atom $j$ and a reference atom $i$, and let $\mathbf{y}_i = \{y_{ij}\}_{j \neq i}$ represent the atomic environment around atom $i$. The total entropy of a structure of this kind is broken down into site entropies in the ACE model, 
\begin{equation} 
S(\mathbf{R}) = \sum_{i=1}^{N} S^* (\mathbf{y}_i),
\end{equation}
where $S^*$ is a site entropy function that depends on its atomic environment $\mathbf{y}_i$. The mapping $S: \mathbb{R}^{3N} \to \R$ is permutation- and isometry-invariant, inherited from the same invariance of the energy. Given a cutoff radius $r_{\text{cut}}$, the ACE site entropy $S^*$ is expressed as
\begin{equation}
S^*(\mathbf{y}_i; \mathbf{c}) = \sum_{B \in \mathcal{B}} c_B B\left(\mathbf{y}_i\right),
\end{equation}
with basis functions \( B \) and parameters $c_B$ that are optimized via a least squares loss minimization. These basis functions are constructed to exhibit invariance under rotations, reflections, and permutations of the atomic environment and are naturally body-ordered, providing a physically interpretable means to converge the fit accuracy. A review of the ACE model and its parameters is given in Appendix \ref{sec:apd:ACE}.

We will explore the effects of estimating the site entropy parameterization $S^*$ from either total entropy, or from site entropies. The latter has the advantage that it increases the amount of available data, but the potential disadvantage that it employs an ad hoc spatial decomposition. (The decomposition used in out proofs may appear natural but in fact there are infinitely many possible alternative decompositions.)

To estimate the coefficients when training total entropy, we require a training dataset which contains a list of atomic configuration $\mathfrak{R}=\{\mathbf{R}\}$ for which the total entropy $\mathcal{S}_\mathbf{R} \in \R$, and entropy derivatives $\nabla \mathcal{S}_\mathbf{R} \in \R^{N \times 3}$ (where $N$ is the total number of atoms in each configuration $\mathbf{R}$) have been evaluated. A possible way to estimate the parameters, closely mimicking parameter estimation for interatomic potentials, is to minimize the quadratic loss function
\begin{eqnarray} \label{eq:quad-loss}
    \mathcal{L}(\mathbf{c}) := \sum_{\mathbf{R} \in \mathfrak{R}} \left( \omega^{2}_{S,\mathbf{R}} | S(\mathbf{c}; \mathbf{R}) - \mathcal{S}_\mathbf{R} |^{2} + \omega^{2}_{\nabla S,\mathbf{R}} | \nabla S(\mathbf{c}; \mathbf{R}) - \nabla \mathcal{S}_\mathbf{R} |^{2} \right),
\end{eqnarray}
with weights $\omega_{S}$ and $\omega_{\nabla S}$ adjusting the significance of the contributions of entropy and its derivatives.

Similarly, to estimate the coefficients when training site entropy on a training dataset $\mathfrak{R}$, we can minimize the following quadratic loss function
\begin{equation} \label{eq:quad-loss-site}
\mathcal{L}^*(\mathbf{c}) := \sum_{\mathbf{R} \in \mathfrak{R}} \left( \omega^{2}_{S^*,\mathbf{R}} \sum_{i=1}^{N} \left( S^*(\mathbf{y}_i; \mathbf{c}) - \mathcal{S}_{i,\mathbf{R}} \right)^2 + \omega^{2}_{G^*,\mathbf{R}} \sum_{i=1}^{N} \sum_{j=1}^{3} \sum_{k=1}^{N} \left( G^*_{ijk}(\mathbf{c}; \mathbf{R}) - \mathcal{G}_{ijk,\mathbf{R}} \right)^2 \right),
\end{equation}
where $\mathcal{S}_{i, \mathbf{R}} \in \R$ represents the site entropy evaluated at site $i$ and $G^*(\mathbf{c}; \mathbf{R}) \in \R^{N \times 3 \times N} $ represents the three-dimensional array with each element $G^*_{ijk}$ representing the derivative of the site entropy at site $i$ with respect to the $j$-th spatial dimension of the atom at the $k$-th site. This tensor captures how the contribution of each atom to the overall entropy varies with changes in the positions of all atoms within the system. Similarly, $\mathcal{G}_{ijk,\mathbf{R}}$ denotes the elements of the reference gradient tensor for the site entropy associated with the structure $\mathbf{R}$. The indices $i$, $j$, and $k$ serve the same purposes as described for $G^*_{ijk}(\mathbf{c}; \mathbf{R})$.

Selecting the appropriate training data, loss functions, and weights, as indicated in the loss function above, is crucial for developing precise models capable of accurate predictions, which will be specified in the following presentation. 

% \co{\sout{When we  can follow the same procedure to fit the parameters against site entropies.}} \cco{I think that needs to be written out in full later.}

\subsection{Fitting Entropy}

In this section, we will present and discuss the fitting results achieved by using the above mentioned parameterisation of entropy via ACE. We will fit models on both total and site entropies and their respective derivatives. The details on the training and testing procedures will be given below.

\subsubsection{Models setup} 

We start by defining the training set, denoted by $\mathfrak{R}$, which encompasses the collection of training data. Geometry optimization is initially performed on the training domain to find the equilibrium atomic positions. Subsequently, with a chosen parameter $\alpha$ called the rattling parameter, and the number of configurations in $\mathfrak{R}$, denoted as $N_{\text{train}} \equiv |\mathfrak{R}|$, atomic positions are perturbed. Each atom is displaced by a vector with components that are uniformly distributed random numbers within the interval $[-\frac{\alpha}{2}, \frac{\alpha}{2}]$. 
% I removed this - it should not be done this way, but TBH it makes no difference, let's just sweep it under the carpet.
% , then scaled to ensure that the displacement is isotropic and does not exceed the maximum magnitude determined by $\alpha$. 
This procedure is repeated $N_{\text{train}}$ times to generate a diverse set of configurations.
A similar approach is utilized to generate a test set. The number of configurations for both the training and test sets, $N_{\text{train}}$ and $N_{\text{test}}$, will be detailed for each case study. 
We choose the weights as $\omega_S \gg \omega_{\nabla S}$ to enforce more accuracy on entropy. 
The parameters $\{c_B\}$ are determined by minimizing the loss functions ~\eqref{eq:quad-loss} and ~\eqref{eq:quad-loss-site} using a Bayesian Ridge Regression (BLR)~\cite{witt2023otentials} solver, which is capable of autonomously determining the importance of input features within the model. 
The training process utilized open-source Julia packages: {\tt JuLIP.jl} package for the creation of test and training datasets, as outlined in \cite{gitJuLIP}, and {\tt ACEpotentials.jl} package, detailed in \cite{witt2023otentials} and accessible at \cite{gitACEPot}, for both ACE basis construction and ACE model fitting.

The hyperparamter choices and results are as follows:
\begin{itemize}
    \item \textbf{Training and Validation Sets:} The dataset includes 100 configurations of rattled Silicon, where each configuration is perturbed by a rattling parameter calculated as $0.1\, \text{\AA} \times U(0,1)$. Here, $U(0,1)$ denotes a random number drawn from a uniform distribution across the interval $(0,1)$. These configurations are encapsulated within a supercell with dimensions $2a_0 \times 2a_0 \times a_0$ and comprise 32 atoms each, with one vacancy present in every configuration. This set was then split into a training data set including 70 configurations and a validation set including 30 configurations.
    
    \item \textbf{Test Set:} The evaluation utilized a test set composed of 100 Silicon configurations, including two groups. The first group contains 50 bulk Silicon configurations, each rattled with a parameter of $0.08\, \text{\AA} \times U(0,1)$. The second group consists of 50 configurations, each including a vacancy and rattled at a parameter of $0.1\, \text{\AA} \times U(0,1)$, refering to the same uniform distribution. The dimensions of the supercell for the test set are identical to those of the training set, maintaining the size of $2a_0 \times 2a_0 \times a_0$.

  \item \textbf{Reference Entropy Calculation:} For the purpose of training, the hessian matrices needed to compute site and total entropy and their respective derivatives were computed using the Stillinger-Weber potential \cite{Stillinger1985Computer} implemented in the {\tt JuLIP.jl} package \cite{gitJuLIP}.

  \item \textbf{Hyperparameter Tuning:} 
  During the fitting process for total entropy, a model with a body order of 3, a maximum polynomial degree of 16, and a cutoff radius of $5.0\,\text{\AA}$ was employed. For site entropy fitting, the chosen model was characterized by a body order of 4, while retaining the same polynomial degree and cutoff radius as used for total entropy fitting. The decision to employ different body orders for each model arose from extensive trial and error. To achieve this, we created 10 different sets including 100 configurations, split into a training set including 70 and a validation set including 30 configurations as described above. We used the training set to fit the models. To find the best combination of body order and degree, we employed a grid search type method and created a grid of body orders (3 and 4) and polynomial degree (8, 10, 12, 14, 16) and evaluated the validation sets' error on each combination to identify the best setup. After this, a final training set was created and used to train the model that was evaluated on the test sets. This indicates that the hyperparameters were carefully tuned to optimize the model's performance for each type of entropy fitting.

The respective weights were chosen as $\omega_{S} = 60$ and $\omega_{\nabla S} = 1$ when training total entropy, and $\omega_{S_{\ell}} = 600$ and $\omega_{\nabla S_{\ell}} = 1$ when training site entropy. The slightly higher choice of weights for site entropy is a result of the fact that, when training site entropy, we have $N$ times more derivative data compared to training total entropy, assuming the same number of training configurations are used for both.

\end{itemize}

The fitting results for both total and site entropies are presented in Figure~\ref{fig:total_S} and Figure~\ref{fig:site_S}. 
For the model fitted to total entropies, the RMSE for the bulk Silicon test set was $1.68 \times 10^{-4}\, k_B$ for entropy and $1.46 \times 10^{-3}\, k_B/\text{\AA}$ for its derivatives. For the configurations with a vacancy, the RMSEs were $3.81 \times 10^{-4}\, k_B$ for entropy and $2.56 \times 10^{-3}\, k_B/\text{\AA}$ for the derivatives.
The RMSE for the model fitted to site entropies for the bulk Silicon test set were $2.58 \times 10^{-4}\, k_B$, and for the derivatives, $2.47 \times 10^{-3}\, k_B/\text{\AA}$ and for the test set including a vacancy the RMSEs were $6.71 \times 10^{-4}\, k_B$, and for the derivatives, $3.95 \times 10^{-3}\, k_B/\text{\AA}$.

\begin{figure}[htb!]
     \centering
     \begin{subfigure}[b]{0.42\textwidth}
         \centering
\includegraphics[width=\textwidth]{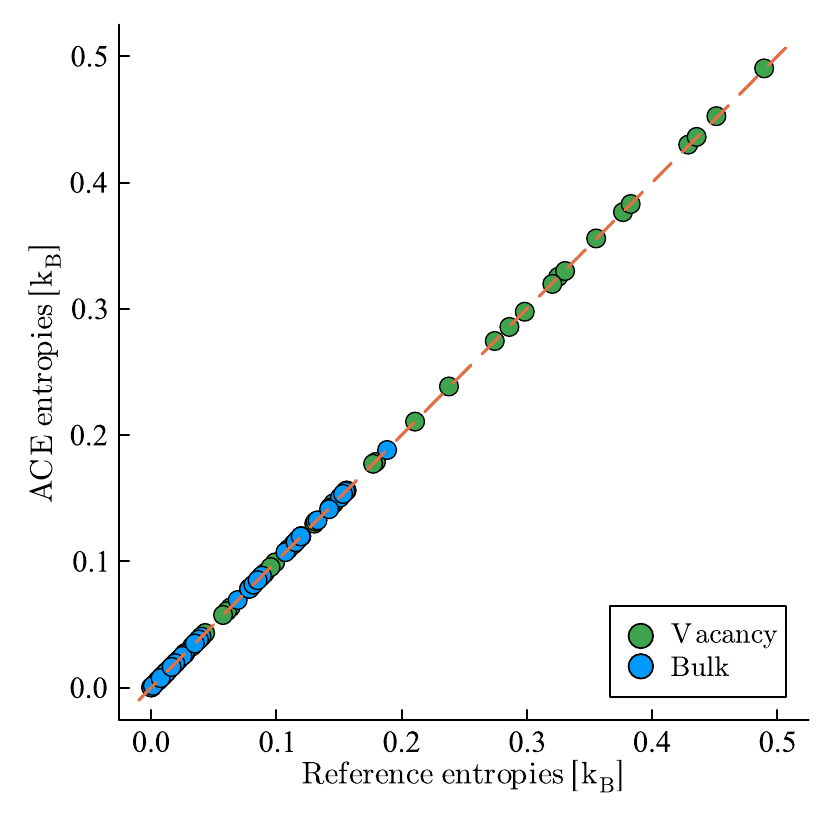} 
     \end{subfigure}
     \begin{subfigure}[b]{0.42\textwidth}
         \centering
         \includegraphics[width=\textwidth]{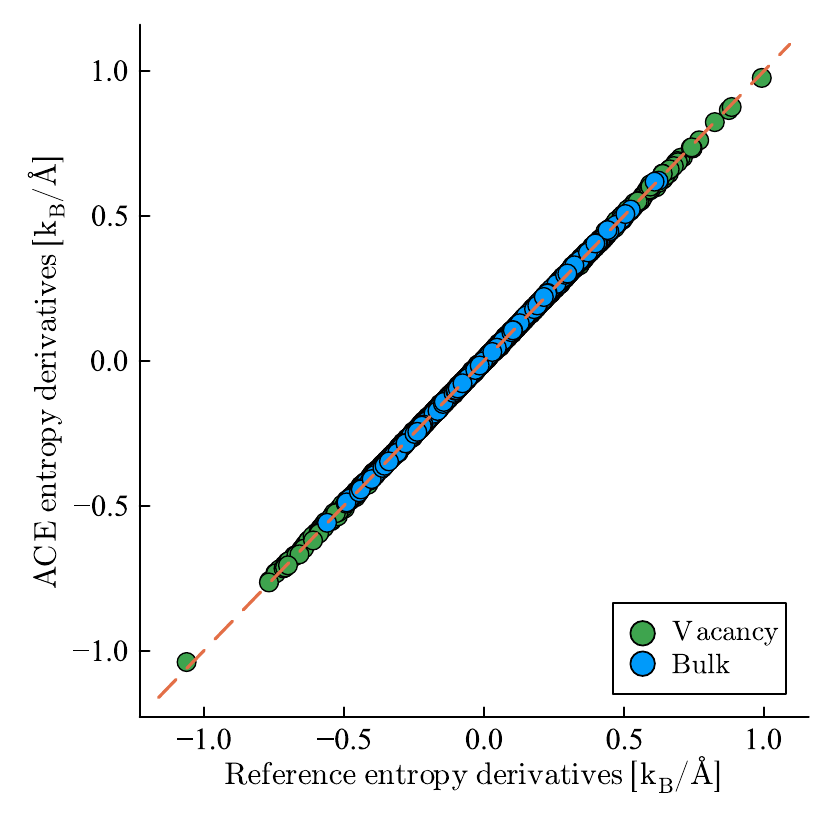}
     \end{subfigure}
    \caption{Total entropy trained on 70 configurations including 31 Silicon atoms and tested on 50 rattled configurations of bulk Silicon with a $2a_0 \times 2a_0 \times a_0$ supercell and and 50 rattled configurations of Silicon including a vacancy.}
    \label{fig:total_S}
\end{figure}

\begin{figure}[htb!]
     \centering
     \begin{subfigure}[b]{0.42\textwidth}
         \centering
\includegraphics[width=\textwidth]{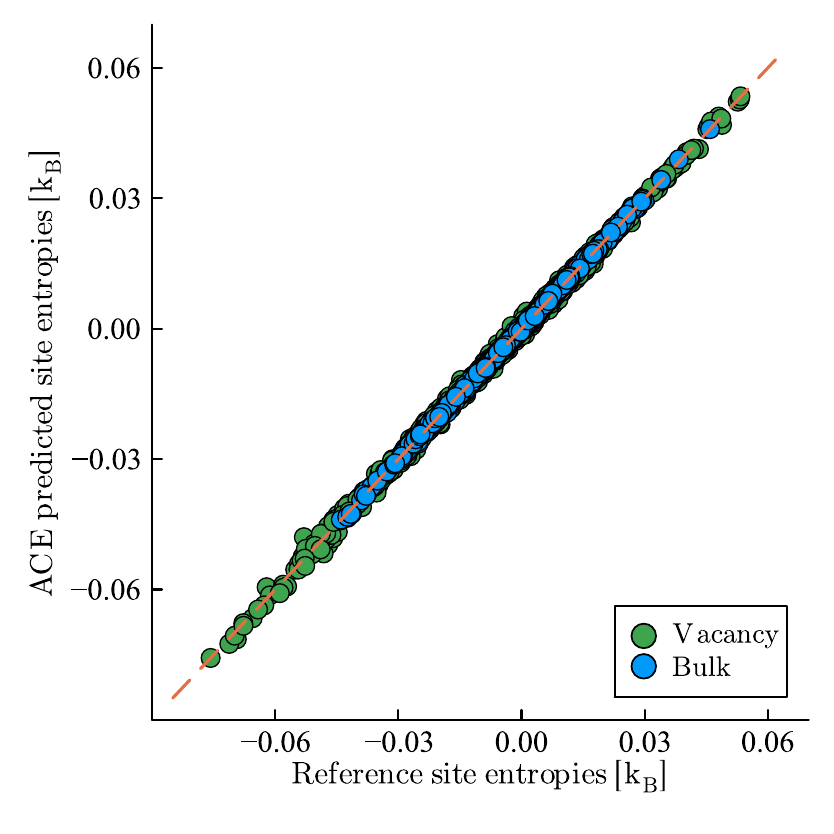} 
     \end{subfigure}
     \begin{subfigure}[b]{0.42\textwidth}
         \centering
         \includegraphics[width=\textwidth]{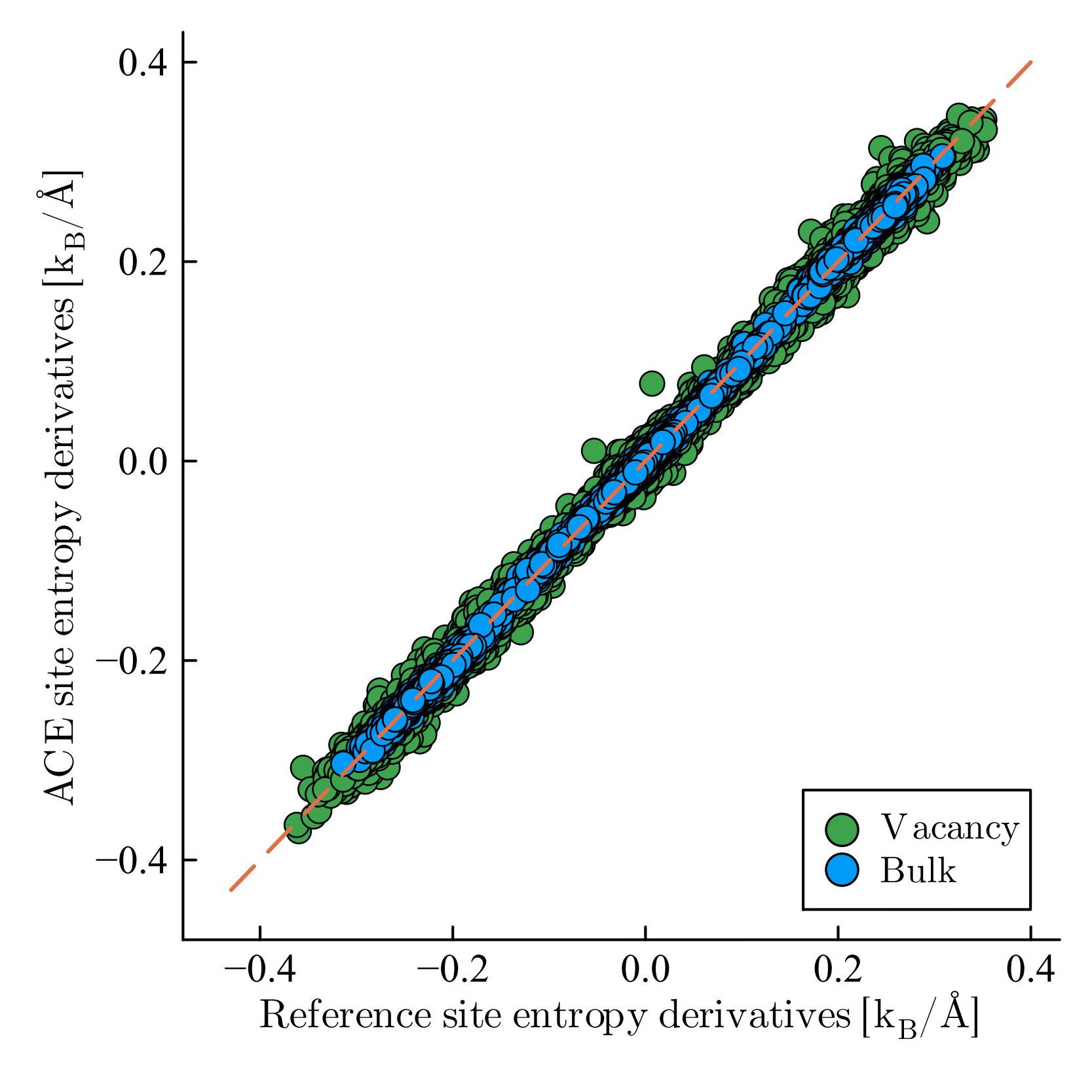}
     \end{subfigure}
    \caption{Site entropy trained on 70 configurations including 31 Silicon atoms and tested on 50 rattled configurations of bulk Silicon with a $2a_0 \times 2a_0 \times a_0$ supercell and and 50 rattled configurations of Silicon including a vacancy.}
    \label{fig:site_S}
\end{figure}

Additionally, we present a comparison between prediciting total entropy by the means of a model fitted only to total and another one only to site entropies. Our training set in both cases included 50 rattled configurations of Silicon (rattling parameter $= 0.1\, \text{\AA} \times U(0,1)$) including a vacancy within a $2a_0 \times 2a_0 \times a_0$ supercell. Our test set included 50 rattled configurations of bulk Silicon (rattling parameter $= 0.08\, \text{\AA} \times U(0,1)$). During the training for total entropy, we constructed a model with a body order of 3, a maximum polynomial degree of 14, and a cutoff radius of $5.0$\AA. For the site entropy, we assembeled a model with a body order of 3, a maximum polynomial degree of 14, and a cutoff radius of $5.0$\AA. The fitting results are showcased in Figure~\ref{fig:compare_S}. Consistent with expectations, the task of fitting site entropies presents greater complexity relative to the fitting of total entropies. For the specific objective of predicting total entropies, it is empirically more precise to utilize a model that has been exclusively trained on total entropy data. This empirical observation strongly suggests that there exists an alternative spatial decomposition of total entropy that is ``easier to fit'' (e.g. more local, lower body-order, smoother) than the explicit construction we employed in our analysis.

% \cco{it isn't specified anywhere how to train against site entropies} \cctina{you mean like in section 4.1? I'm a bit confused on what you mean by How.} \cco{you should write down the loss that you used to fit to site entropies.} \cctina{done above!}

\begin{figure}[!htbp]
    \includegraphics[height=7.2cm]{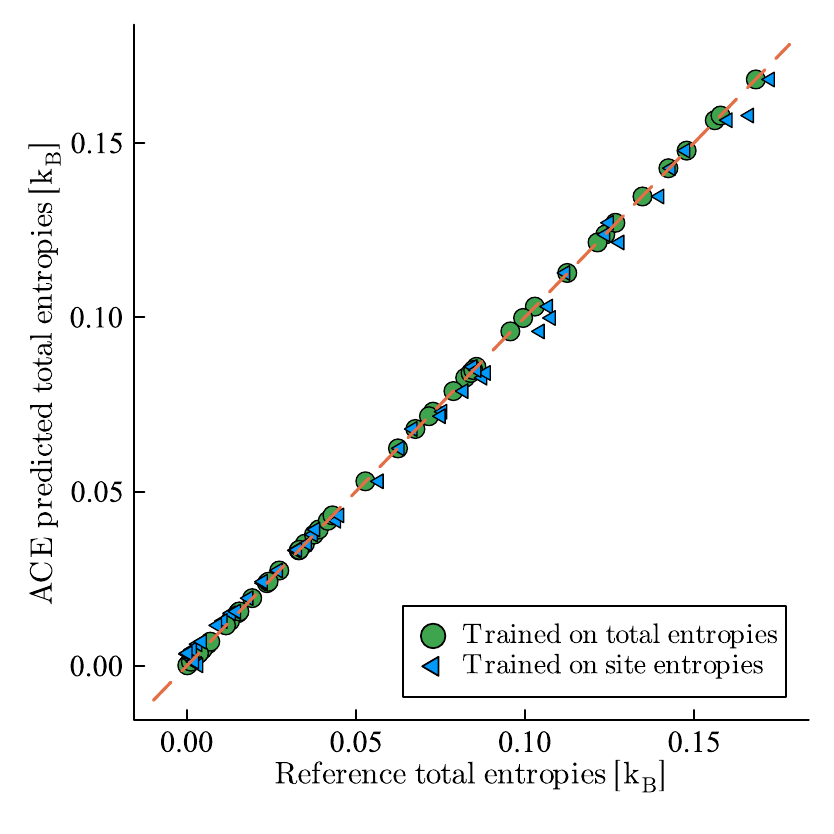}
    \caption{Comparing the performance of two models, one fitted exclusively to total entropy and one fitted solely to site entropy}
    \label{fig:compare_S}
\end{figure}

\subsection{Learning curves} \label{LR}

We train a series of models with identical hyperparameters except for $r_{\text{cut}}$. The body order and polynomial degree of the trained models  were 3 and 14 for total and 3 and 16 for site entropy. Training set sizes are varied within the range of 15 to 200 configurations, ensuring that each set comprises an equal configurations of the four distinct Silicon configuration sets all within a supercell of dimensions \( 2a_0 \times 2a_0 \times a_0 \): (i) Rattled ($\alpha = 0.1 \times U(0,1)$) bulk Silicon containing 32 atoms, (ii) Rattled ($\alpha = 0.1 \times U(0,1)$) Silicon supercell including a single vacancy, (iii) Rattled ($\alpha = 0.1 \times U(0,1)$) Silicon supercell including a divacancy , and (iv) Rattled ($\alpha = 0.1 \times U(0,1)$) Silicon supercell including an interstitial. The test set included a total of 80 configurations, with 20 configurations from each aforementioned sets. 

% \cco{it is not clear why you are changing the protocol} \cctina{I don't know if you remember but I had to play around a lot to get the good LR curves. I didn't go back to train the models in prev subsection using the same hyper-parameters I got the LR curves with, but I can do that. Should I?} \cco{just needs to be made clear in the text} \cctina{fixed!}

Figure~\ref{fig:LC_tot} and Figure~\ref{fig:LC_site} display the log-log plots corresponding to the total and site entropies, respectively. It is noteworthy that upon comparison, we observe that the convergence happens relatively fast for both cases, which is not unexpected due to the simplicity of our training sets. It is also evident that the total entropy values reach convergence earlier when compared to the site entropies. This observation also aligns with expectations, attributed to the inherently lower complexity associated with fitting total entropy. 

To conduct a quantitative analysis on model performance regarding the prediction of total entropies, with fitting results shown in Figure~\ref{fig:compare_S}, we trained two distinct models: one fitted exclusively against total entropies, and the other against site entropies. Our first objective was to investigate how variations in the number of training configurations influenced the RMSE observed on the test set. We used training sets of varying sizes, containing rattled configurations of Silicon with a vacancy (with a rattling parameter of $0.1 \text{\AA} \times U(0,1)$) within a $2a_0 \times 2a_0 \times a_0$ supercell. We used a test set consisting of 50 rattled bulk Silicon configurations (rattling parameter being $0.08 \text{\AA} \times U(0,1)$). The training was done using a body order of 3 and a polynomial degree of 14.
The fitting results, corresponding to two differing cut-off radii, $5.0 \text{\AA}$ and $7.0 \text{\AA}$, are depicted in Figure~\ref{fig:compare_totA}. Furthermore, Figure~\ref{fig:compare_totB} showcases the convergence behavior of these models as we increase the basis size, while maintaining a constant cut-off radius of $5.0 \text{\AA}$ and using 50 configurations for training. A convergence trend is observed in both models with increasing basis size, in line with our expectations.

% \begin{figure}[htb!]
%      \centering
%      \begin{subfigure}[b]{0.42\textwidth}
%          \centering
%         \includegraphics[width=\textwidth]{figs/LR_flattening_total.pdf}
%          \caption{Total entropy}
%      \end{subfigure}
%      \begin{subfigure}[b]{0.42\textwidth}
%          \centering
%          \includegraphics[width=\textwidth]{figs/LR_flattening_total_forces.pdf}
%          \caption{Total entropy derivatives}
%      \end{subfigure}
.      %     \caption{RMSE Convergence for total entropy varying $r_{cut}$ equal to $4.0 \, \text{\AA}$, $5.0 \, \text{\AA}$, $6.0 \, \text{\AA}$, and $7.0 \, \text{\AA}$ for Si. Each point corresponds to $N$ observations corresponding to total entropies and $96 \times N$ observations corresponding to total entropy derivatives where $N$ is equal to the number of training configurations. } 

%     \label{fig:LC_tot}
% \end{figure}

\begin{figure}[htb!]
     \centering
     \begin{subfigure}[b]{0.42\textwidth}
         \centering
        \includegraphics[width=\textwidth]{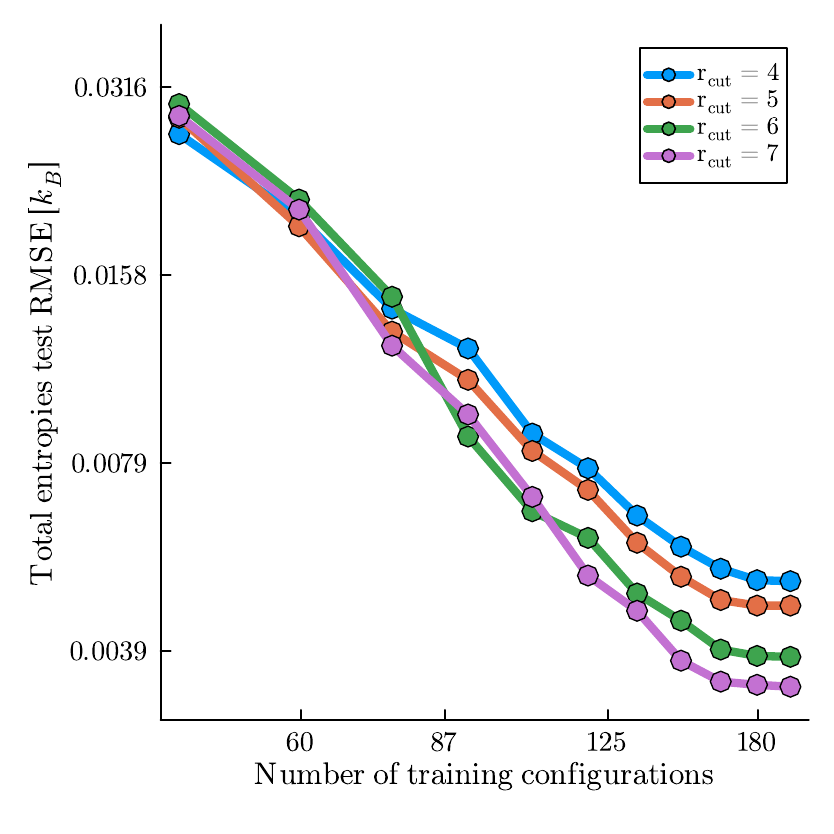}
         \caption{Total entropy}
     \end{subfigure}
     \begin{subfigure}[b]{0.42\textwidth}
         \centering
         \includegraphics[width=\textwidth]{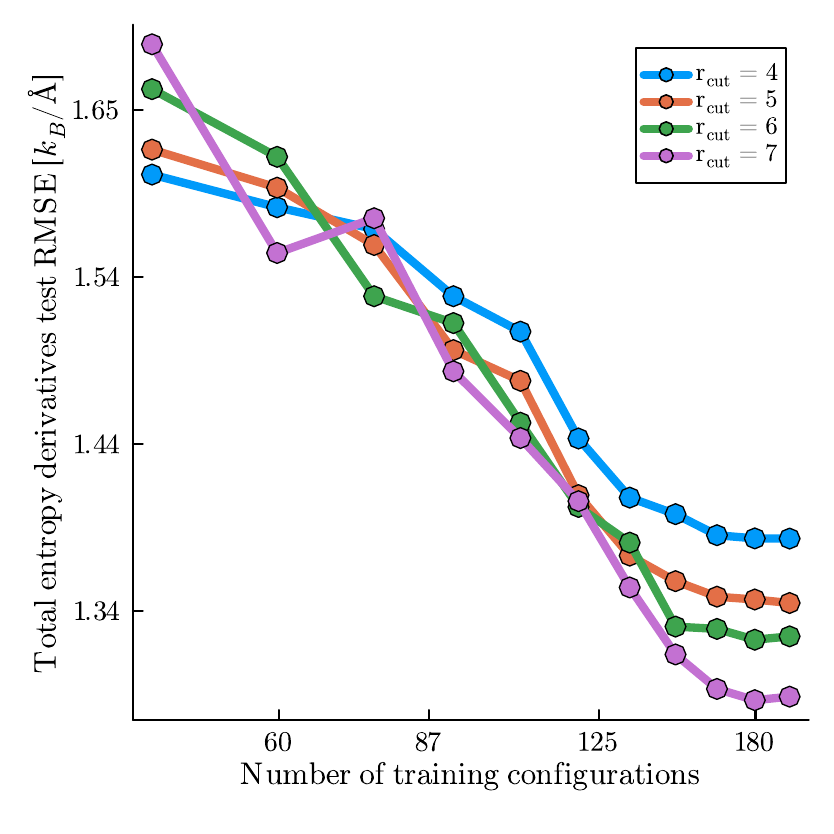}
         \caption{Total entropy derivatives}
     \end{subfigure}
    \caption{RMSE Convergence for total entropy varying $r_{cut}$ equal to $4.0 \, \text{\AA}$, $5.0 \, \text{\AA}$, $6.0 \, \text{\AA}$, and $7.0 \, \text{\AA}$ for Si. Each point corresponds to $N$ observations corresponding to total entropies and $96 \times N$ observations corresponding to total entropy derivatives where $N$ is equal to the number of training configurations.} 

    \label{fig:LC_tot}
\end{figure}

\begin{figure}[htb!]
     \centering
     \begin{subfigure}[b]{0.42\textwidth}
         \centering
        \includegraphics[width=\textwidth]{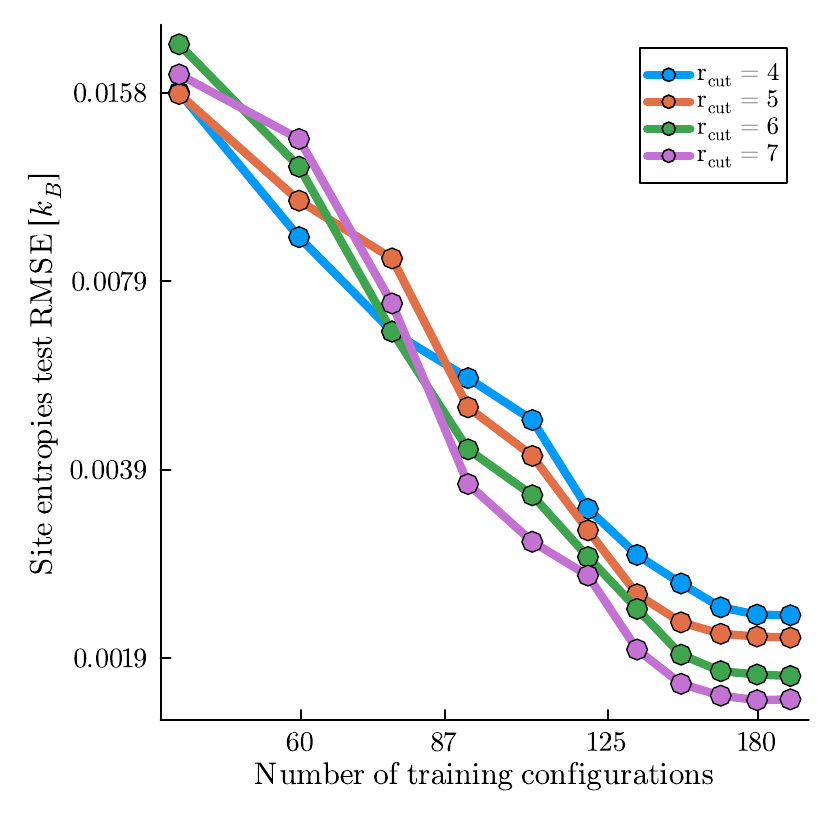}
         \caption{Site entropy}
         \label{eq:site_LR}
     \end{subfigure}
     \begin{subfigure}[b]{0.42\textwidth}
         \centering
         \includegraphics[width=\textwidth]{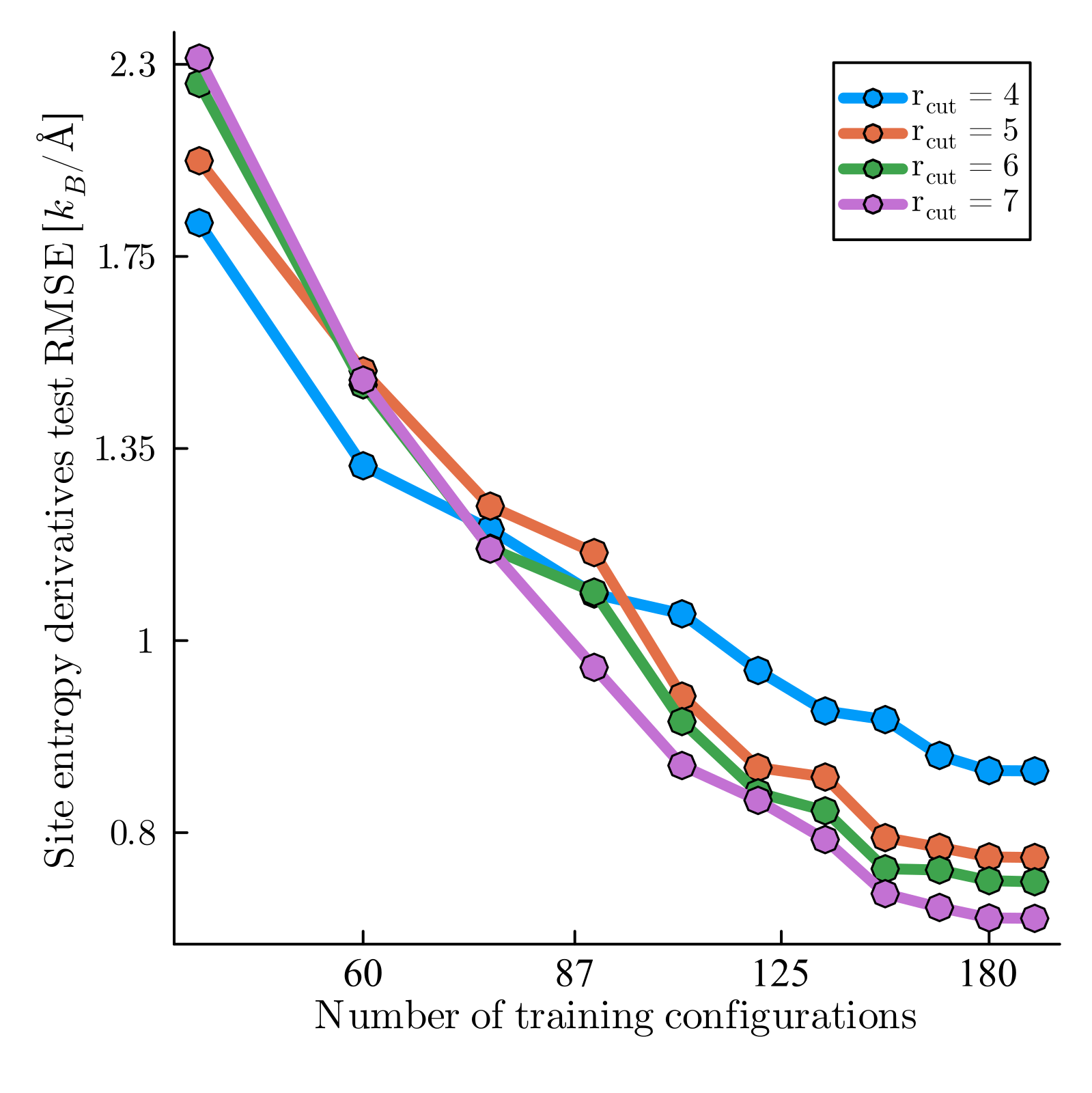}
         \caption{Site entropy derivatives}
     \end{subfigure}
    \caption{RMSE Convergence for site entropy varying $r_{cut}$ equal to $4.0 \, \text{\AA}$, $5.0 \, \text{\AA}$, $6.0 \, \text{\AA}$, and $7.0 \, \text{\AA}$ for Si. Each point corresponds to $32 \times N$ observations corresponding to site entropies and $96 \times 32 \times N$ observations corresponding to site entropy derivatives where $N$ is equal to the number of training configurations.}

        \label{fig:LC_site}
\end{figure}

\begin{figure}[htb!]
     \centering
     \begin{subfigure}[b]{0.42\textwidth}
         \centering
        \includegraphics[width=\textwidth]{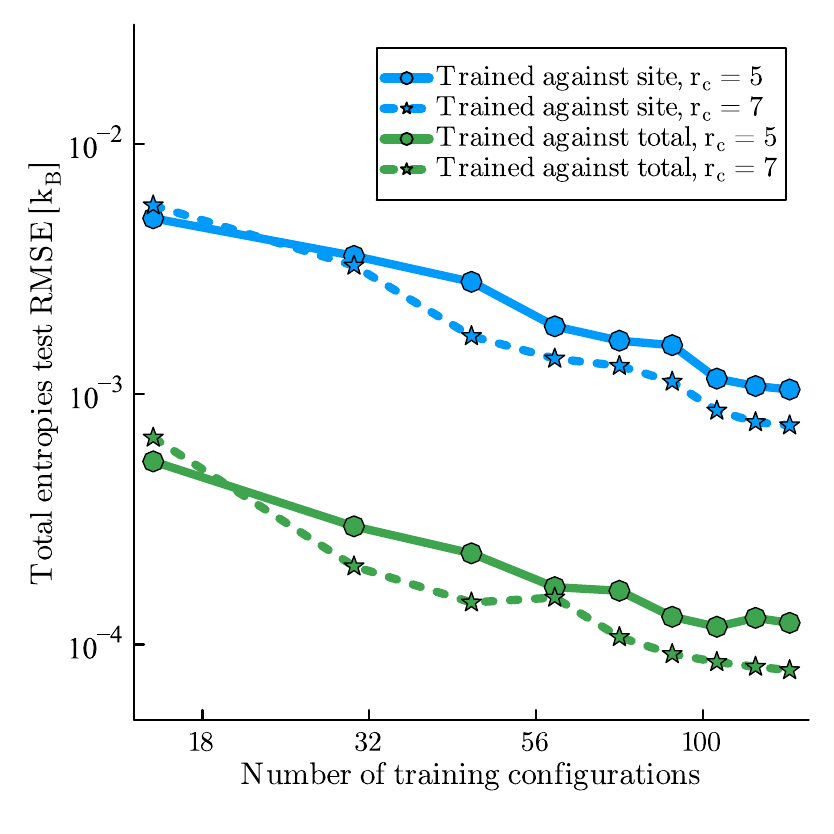}
        \caption{Effect of number of configurations}
         \label{fig:compare_totA}
     \end{subfigure}
     \begin{subfigure}[b]{0.42\textwidth}
         \centering
         \includegraphics[width=\textwidth]{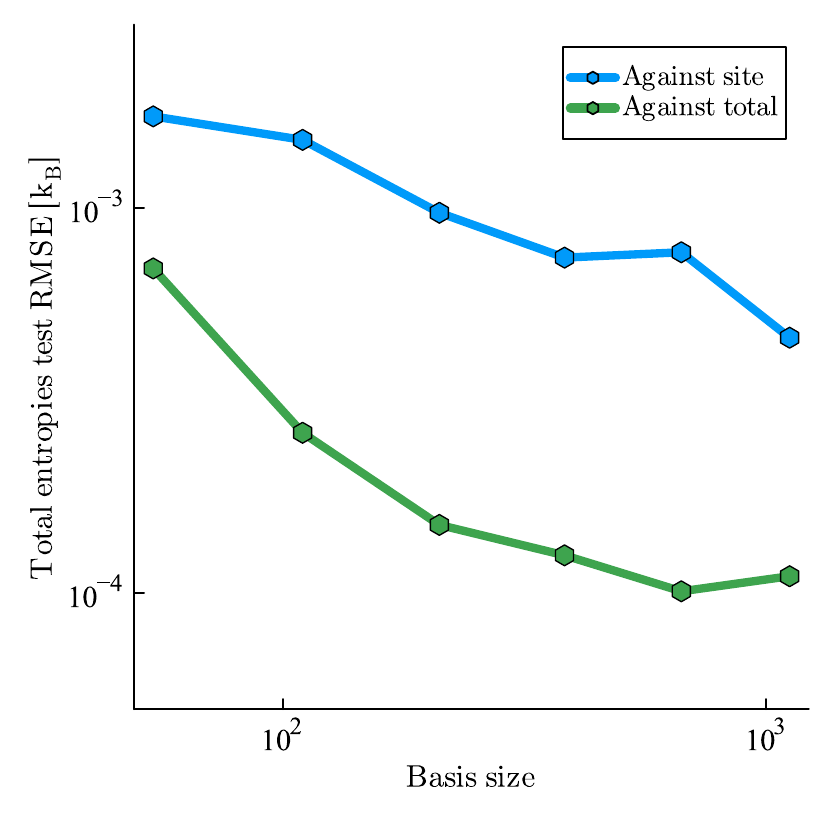}
         \caption{Effect of basis size}
     \label{fig:compare_totB}
     \end{subfigure}
     \caption{RMSE Convergence for total entropy for Si. Figure (A) shows the RMSE error of total entropy evaluated on the test set using two different models. One is solely trained against site entropies (\textbf{blue}) and once is solely trained against total entropies (\textbf{green}) for two different values of cut-off radius, $r_c = 5.0$ \AA \: (\textbf{solid line}) and $r_c = 7.0$\AA \; (\textbf{dashed line}). Figure (B) showcases the RMSE error of total entropy evaluated on the test set for these two models with a fixed cut off radius of $5.0$ \AA \: on a dataset containing 50 training configurations for varying basis size.}
    \label{fig:compare_tot}
\end{figure}

\subsection{Attempt frequency prediction} 
\label{sec:sub:TST}
%
% \sout{
% According to the Transition State Theory (TST)~\cys{refs}, during crystal deformation, thermal excitation allows atoms to transition over energy barriers leading to defect formation. The rate at which this transition occurs, represented as \( K_N \), is determined by evaluating the relative densities of the two equilibrium states. Specifically, in TST:
% \begin{equation}
% K^{\text{TST}}_N = \frac{\int_S e^{-\beta \E_N(u)} du}{\int_A e^{-\beta \E_N(u)} du}
% \end{equation}
% }

% The so-called 'transition state' is a unique point in the potential energy surface that represents the highest energy level along the minimum energy pathway connecting two equilibrium states. In scenarios with a large value of \(\beta\) (corresponding to low temperatures), For sufficiently large $\beta$ (low temperature regime), $\int_S e^{-\beta \E_N(u)} \, du$ and $\int_A e^{-\beta \E_N(u)} \, du$ are concentrated close to $\bar{u}^{\rm saddle}_N$ and  $\bar{u}^{\rm min}_N$ respectively. 
%
% \quad \cco{should we move this paragraph to the relevant application section?}
An important property related to the vibrational entropy is the {\em transition rate}. Within harmonic transition state theory (HTST) it is defined by 
\begin{equation}
\mathcal{K}^{\rm HTST}_N =
    \exp\Big( - \beta \big[
      \F_N(\bar{u}^{\rm saddle}_N) -
      \F_N(\bar{u}^{\rm min}_N) \big] \Big),
      \label{eq:HTST}
\end{equation}
where, for a critical point $\bar{u}$, 
\[
    \F_N(\bar{u})
    = 
    \E_N^{\rm def}(\bar{u}) 
    + \beta^{-1}
      S_N(\bar{u}).
\]
Notably, in materials modeling, especially for systems operating at temperatures significantly below the melting point, the harmonic approximation is generally deemed reliable~\cite{julian2016, kramer}.

Traditionally, $\mathcal{K}^{\rm HTST}_N$ is more typically written as 
\[
    \mathcal{K}^{\rm HTST}_N = 
   \bigg(\frac{
      \textstyle \prod \lambda_j^{\rm min}
   }{
      \textstyle \prod \lambda_j^{\rm saddle}
   }\bigg)^{1/2}
   \, \exp\Big( - \beta \big[
      \E_N(\bar{u}^{\rm saddle}_N) -
      \E_N(\bar{u}^{\rm min}_N) \big]
      \Big),
\]
where the products involve only the positive eigenvalues of, respectively, the hessians $\nabla^2 \E_N(\bar{u}^{\rm min}_N)$ and  $\nabla^2 \E_N(\bar{u}^{\rm saddle}_N)$. However, for our purposes, the formulation \eqref{eq:HTST} is far more convenient. 

To showcase the potential application of our work, we predict the entropic term in the harmonic transition rate \eqref{eq:HTST}. In case of vacancy migration in a crystal, this term is also called the attempt frequency. The system we considered was Copper in a  \( 2a_0 \times 2a_0 \times a_0 \) supercell including a vacancy.  We used an embedded atom model (EAM) potential by Mishin et al. \cite{PhysRevB.63.224106}, as our reference potential and aimed to predict the attempt frequency of a vacancy migrating from one site to the adjacent one. We obtained the index-1 saddle point and minima using the Nudged Elastic Band (NEB) method \cite{10.1063/1.1329672} implemented in Atomic Simulation Environment (ASE) \cite{ase-paper}. We started with a data set including 30 configurations of rattled minima ($\alpha = 0.1 \times U(0,1)$) and predicted the entropy corresponding to the index-1 saddle point. The RMSE for this prediction was $2.341 \; \text{Hz}$. We then added the un-rattled index-1 saddle point to the data set, the error for that prediction is shown in Figure~\ref{fig:TST}. After that, we removed the index-1 saddle point and gradually added rattled saddle points ($\alpha = 0.1 \times U(0,1)$) to the training set and compute the RMSE as demonstrated in Figure~\ref{fig:TST} in a log-log plot. Our test set included 30 rattled saddle points ($\alpha = 0.1 \times U(0,1)$). To train the models, we used a body order of 3 and a polynomial degree of 12. The aim of the gradual adding of rattled saddle points is to add more information about the area between the minima and the saddle point to the model. It is evident that adding more saddle configurations to the training set effects the accuracy of the prediction tremendously.

\begin{figure}[htb!]
\includegraphics[height=7.2cm]{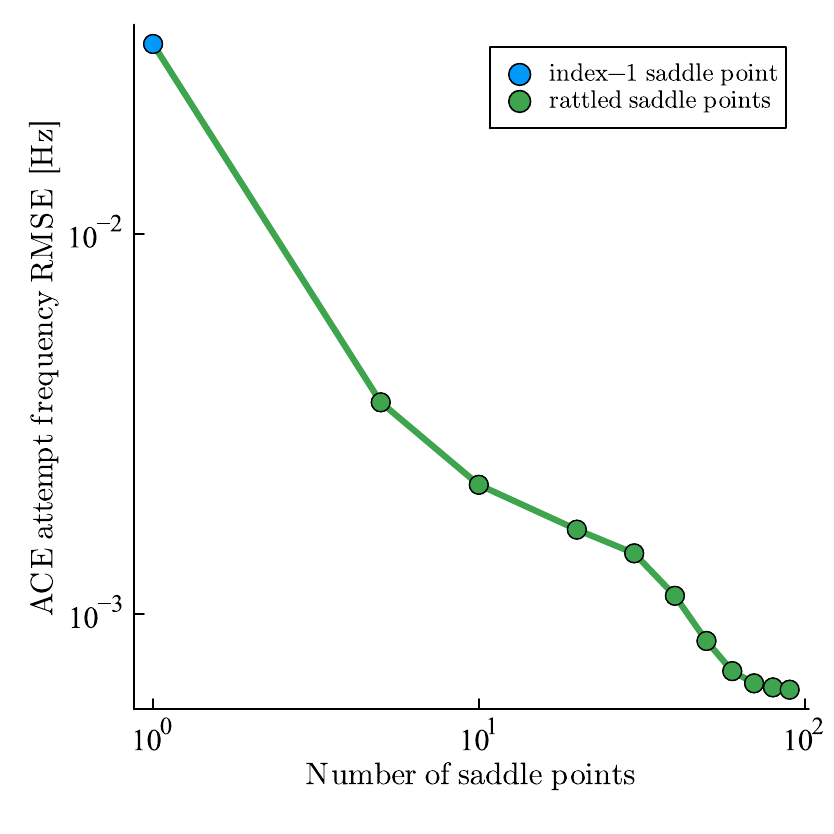}
         \caption{Attempt frequency of a vacancy migrating from one site to the adjacent one in Copper. The index-1 saddle point shows the error for a model trained on 30 rattled minima and the index-1 saddle point.}
    \label{fig:TST}
\end{figure}

\section{Conclusion}
\label{sec:conclusion} 

In this work, we have demonstrated that the total entropy can be decomposed into atomic site contributions and have rigorously estimated the locality of site entropy. Our analysis suggests that vibrational entropy can be accurately predicted using a surrogate model for site entropy, which we have developed using machine learning techniques, specifically employing the Atomic Cluster Expansion (ACE) model. Our numerical experiments primarily focus on point defects such as vacancies and interstitials. We have showcased the robustness of our approach in predicting vibrational entropy and the attempt frequency for transition rates governing point defect migration, which is a critical aspect of transition state theory rate approximations.

While the current study emphasizes the effectiveness of our approach in a relatively straightforward context, it also enables a detailed and rigorous examination, albeit within certain limitations. There are numerous potential generalizations that can be explored, including extending the methodology to more intricate material and defect geometries.

\appendix 

\section{Proofs}\label{sec: appendix_prrof}

In this section, we begin by introducing the necessary concepts and models within the primary context of this work. We accomplish this by conducting a comprehensive review of the framework proposed in~\cite{chen19, bcshap2016, fu2023adaptive}, while also adapting their approaches to align with the specific objectives of our current work. 

\subsection{Preliminaries}
\label{sec: appendix_prelem}

We introduce the semi-discrete Fourier transform (SDFT) 
\begin{equation}
\label{eq: sdFourier}
\hat{u}(k) := \sum_{\ell \in \Lambda} e^{i k \ell} \cdot u(\ell), \quad \text{with inverse} \quad u(\ell) = \frac{1}{|\mathcal{B}|} \int_{\mathcal{B}} e^{-i k \ell} \cdot \hat{u}(k) \,{\rm d}k,
\end{equation}
where $\mathcal{B} = \pi \mA^{-T} (-1,1)^d$ is a fundamental domain of reciprocal space
(equivalent to the first Brillouin zone) and has the volume $\lvert \mathcal{B} \rvert =
\frac{(2\pi)^d}{\lvert \det \mA \rvert}$.

We now review a characterization of \bF, which plays a key role in proving our main locality theory (cf. Theorem~\ref{thm:local}). Since $H^{\rm hom}$ is circulant, we can represent ${\rm \bF}\omega=F*\omega$ and define $F$ via its Fourier transform. 

To that end, we recall that
\begin{eqnarray*}
    \<H^{\rm hom}u, v\> = \sum_{\ell\in\Z^d} \nabla^2 V({\bf 0}) \big[Du(\ell), Dv(\ell)\big],
\end{eqnarray*}
then applying the SDFT we obtain
\begin{eqnarray}
    \<H^{\rm hom}u, v\> = \frac{1}{\mathcal{B}}\int_{\mathcal{B}} \hat{u}(k)^{*} \hat{h}(k) \hat{v}(k) \,{\rm d}k
\end{eqnarray}
with 
\begin{eqnarray*}
    a^{T} \hat{h}(k) b := \nabla^2 V({\bf 0})\big[\big((e^{-ik\rho}-1)a\big)_{\rho\in\mathcal{R}}, \big((e^{ik\rho}-1)b\big)_{\rho\in\mathcal{R}}\big].
\end{eqnarray*}

One can also reduce $\hat{h}(k)$ to the simpler form
\begin{eqnarray*} 
\hat{h}(k)=4\sum_{\rho\in\mathcal{R}'} A_{\rho} \cdot \sin^2\left(\frac{k\cdot \rho}{2}\right)
\end{eqnarray*}
with $\mathcal{R}'=(\mathcal{R}\cup\{0\})+(\mathcal{R}\cup\{0\})$, see~\cite[Section 6.2]{bcshap2016} for more details. Moreover, the stability implies that $c_0|k|^2\cdot\mathsf{I} \leq \hat{h}(k) \leq c_1|k|^2\cdot\mathsf{I}$ with the identity matrix $\mathsf{I}$ for all $k\in\mathcal{B}$~\cite{Hudson:stab}. We observe that $|\hat{h}(k)^{-1/2}|\lesssim|k|^{-1}$ as $|k|\rightarrow0$, hence we can define
\begin{align} 
   F(\ell) &:= \frac{1}{|\mathcal{B}|} \int_{\mathcal{B}} e^{-i k \ell} \cdot \hat{F}(k) \,{\rm d}k,
   \qquad \text{where} \quad \hat{F}(k) = \hat{h}(k)^{-1/2}, \label{eq:defn_F} 
   \\
   (\pmb{\rm F}u)(\ell) &:=
      \sum_{m \in \L} \big( F(\ell-m) - F(-m) \big) u(m).
\label{eq:defn_Fop}
\end{align}
The constant shift $\sum_{m}F(-m)u(m)$ in the definition of \bF$u$ ensures that \bF$u$ is well-defined~\cite{julian2016}.

The following lemma presents key properties of the operator ${\rm \textbf{F}}$, which holds a pivotal role in our framework. We state it here to ensure the completeness of our presentation.

\begin{lemma} {\cite[Lemma 3.1]{julian2016}} \label{th:properties_F}
    Let $F : \L \to \R^{m \times m}$ and {\rm \textbf{F}} be defined by \eqref{eq:defn_F} and \eqref{eq:defn_Fop}, respectively. Then, we have: 
(i) for any $\boldsymbol{\rho} \in \mathcal{R}^j$, $j \geq 0$, there exists a constant $C$ such that 
   \begin{equation*}
      |D_{\boldsymbol{\rho}} F(\ell)|\leq C (1+|\ell|)^{1-d-j}\quad \forall \ell\in\Z^d.
   \end{equation*}
(ii) {\rm \bF} $\in \mathcal{L}\big(\ell^2, \dot{\mathcal{W}}^{1,2}\big)$.
(iii) {\rm \bF}$^*  H^{\rm hom}${\rm \bF} $= I$, understood as operators $\ell^2 \to \ell^2$.
\end{lemma} 
\\

A crucial quantity for proving locality is the estimation of the resolvent $\mathscr{R}_z(u) = (zI- \pmb{\rm F}^{*} H(u) \pmb{\rm F})^{-1}$. Before providing this estimate, we first review a useful lemma.

\begin{lemma} {\cite[Lemma 3.3 (i)]{julian2016}}
Let $X$ be a Hilbert space. Let $A \in \mathcal{L}(X, X)$ be a bounded linear operator with range of finite dimension at most $r \in \N$ and $I +A$ is invertible, then there exist
$c_1(A), \dots, c_{r+1}(A) \in \R$ such that
\begin{equation}
\label{eq: finite-rank correction}
   (I+A)^{-1}=I + \sum_{j=1}^{r+1} c_j(A) A^j.
\end{equation}
\end{lemma}

As discussed in Section~\ref{sec:sub:tlim}, we are ready to give the estimate of resolvent $\mathscr{R}_z(u)$ for $u\in\mathcal{U}$, where 
\[
   \mathcal{U}(\underline\sigma, \overline\sigma) :=
      \big\{ u \in \dot{\mathcal{W}}^{1,2}:       \sigma\big[ \pmb{\rm F}^{*}  H(u) \pmb{\rm F} \big] \cap (0, \infty) \subset [\underline\sigma, \overline\sigma] \big\},
\]
with $0 < \underline\sigma < \overline\sigma$ fixed throughout this work and all constants in the following are allowed to depend on them.

\begin{lemma} \label{th:resolvent estimate}
Let $\mathscr{R}_z(u)$ be the resolvent operator defined by \eqref{eq:resolvent}. Then for all \( u \in \mathcal{U} \) there exists a constant \( C_1 > 0 \) independent of \( m \) or \( n \), such that
\begin{equation}  \label{eq:limit:resolvent estimate vNper}
    \big| [ \mathscr{R}_z(u) ]_{m n} \big|
    \leq C_1 \big(1+|r_{mn}|\big)^{-d}.
\end{equation}
\end{lemma}

\begin{proof}

The main technique we use is to decompose a Hessian operator $H$ into two components $H=H^{\rm d}+ H^{\rm h}$ where $H^{\rm d}$ as finite rank and represents the defect core while $H^{\rm h}$ is close to $H^\text{hom}$ and thus represents the far-field. 

We begin by splitting the difference of the hessians $H(u)-H^{\text{hom}}$ into a sum of a finite rank operator representing the defect core and an infinite rank part representing the far field. To that end, given $M$ sufficiently large, we denote
\begin{equation} \label{eq:res:defn of HM}
   \< H^M(u) v, z\>
   := \sum_{\lvert \ell \rvert \leq M} \nabla^2V({\bf 0})\big[Dv(\ell), Dz(\ell)\big]
      +\sum_{\lvert \ell \rvert > M} \nabla^2V(Du(\ell))\big[Dv(\ell), Dz(\ell)\big].
\end{equation}
The resolvent $\mathscr{R}_z$ can be decomposed as follows:
\begin{equation} 
\begin{aligned}
\mathscr{R}_z = (zI- \pmb{\rm F}^{*}  H(u) \pmb{\rm F})^{-1}  =& ~ \Big[ zI- \pmb{\rm F}^{*}  H^M(u) \pmb{\rm F} + \pmb{\rm F}^{*}  H^M(u) \pmb{\rm F} - \pmb{\rm F}^{*}  H(u) \pmb{\rm F}\Big]^{-1} \\
    =&~ \Big[\big(zI- \pmb{\rm F}^{*}  H^M(u) \pmb{\rm F} \big) + \pmb{\rm F}^{*} \big(H^M(u)-H(u)\big) \pmb{\rm F}\Big]^{-1} \\
    =&~ \Big[(\mathscr{R}_z^M)^{-1}  + \pmb{\rm F}^{*} \big(H^M(u)-H(u)\big) \pmb{\rm F}\Big]^{-1} \\
    =&~ \Big[(\mathscr{R}_z^M)^{-1} \big ( I + \mathscr{R}_z^M \pmb{\rm F}^{*} \big(H^M(u)-H(u)\big) \pmb{\rm F}\big ) \Big]^{-1} \\
    =&~ \Big[I + \mathscr{R}_z^M \pmb{\rm F}^{*} \big(H^M(u)-H(u)\big) \pmb{\rm F}\Big]^{-1} \cdot \mathscr{R}_z^M\\
    :=&~ \big(I + B^M \big)^{-1} \cdot \mathscr{R}_z^M,
\end{aligned}
\end{equation}
where $\mathscr{R}_z^M = \big(zI- \pmb{\rm F}^{*}  H^M(u) \pmb{\rm F}\big)^{-1}$.

Since the resolvent $\mathscr{R}_z(u)$ exists for all $z \in \mathbb{C}
, u \in \mathcal{U}$, the inverse on the right hand side exists as well.
Additionally, the term $B^M:=\mathscr{R}_z^M \pmb{\rm F}^{*} \big(H^M(u)-H(u)\big){\bf F}$ has finite-dimensional range since this is clearly the case for $H^M(u) - H(u)$. Thus, according to Lemma~\ref{eq: finite-rank correction} it follows that
\begin{equation*} 
   \big(I + B^{M} \big)^{-1} =
   I + \sum_{j=1}^{r+1}  c_j(B^{M}) \big(B^{M}\big)^{j}.
\end{equation*}
The constants $c_j$ remain uniformly bounded in $z, u \in \mathcal{U}$, therefore we only need to estimate
\begin{equation} 
\big(B^{M}\big)^{j} \mathscr{R}_z^M = \big(\mathscr{R}_z^M \pmb{\rm F}^{*} \big(H^M(u)-H(u)\big) \pmb{\rm F} \big)^j \mathscr{R}_z^M,
\end{equation}
for $1 \leq j \leq r+1$. 
Using the techniques shown in \cite[Eq. (7.13)]{julian2016}, one can obtain
\begin{equation}
\begin{aligned}
\left| \Big(\pmb{\rm F}^{*} \big(H^M(u)-H(u)\big) \pmb{\rm F} \Big)_{mn} \right| \lesssim& \sum_{\lvert \ell \rvert \leq M} (\lvert \ell \rvert +1 )^{-d} (\lvert \ell -m \rvert +1) ^{-d} (\lvert \ell-n \rvert +1 )^{-d}  \\
\leq& \sum_{\ell} (\lvert \ell \rvert +1 )^{-d} (\lvert \ell -m \rvert +1) ^{-d} (\lvert \ell-n \rvert +1 )^{-d} \\
\lesssim& (|n|+1)^{-d}(|n-m| + 1)^{-d}+(|m| + 1)^{-d}(|n-m| +1)^{-d} \\ 
 &+(|n|+1)^{-d}(|m|+1)^{-d}.
\end{aligned}
\end{equation}
% where $\mathcal{L}_{1}(m,n) := (|n|+1)^{-d}(|n-m| + 1)^{-d}+(|m| + 1)^{-d}(|n-m| +1)^{-d}+(|n|+1)^{-d}(|m|+1)^{-d}$.
% Furthermore, according to {\cite[Eq. (3.26)]{julian2016}}
% \begin{equation} 
%    \lvert (\mathscr{R}_z^M(u)- \mathscr{R}_z^{\text{hom}})_{mn} \rvert
%     \leq \mathcal{L}_{1}(m,n).
% \end{equation}
Hence, applying {\cite[Eq. (3.19) and Eq. (3.26)]{julian2016}}, we have
\begin{equation*}
\big| \big[ \mathscr{R}_z \big]_{m n} \big| = \left| \big [\big(\mathscr{R}_z^M \pmb{\rm F}^{*} \big(H^M(u)-H(u)\big) \pmb{\rm F} \big)^j \mathscr{R}_z^M \big ]_{mn} \right| \leq C_1 (1+|r_{mn}|)^{-d},
\end{equation*}
where $C_1>0$ is independent of $m$ or $n$. This completes the proof.
\end{proof}

We are ready to give the proof of our main theorem (cf.~Theorem~\ref{thm:local}).

\begin{proof}[Proof of Theorem~\ref{thm:local}]
The site entropy $\mathcal{S}_{\ell}(u)$ can be expressed using the contour integral defined in~\eqref{eq:local_S+},
\begin{equation}
   \mathcal{S}_\ell(u) :=  - \frac{1}{2}\frac{1}{2\pi i} {\rm Trace} \left[ \oint_{\mathcal{C}} \log^+(z)\cdot\big(zI - \mathbf{F}^{*} H(u) \mathbf{F}\big)^{-1}\,{\rm d}z \right]_{\ell\ell},
\end{equation}
where \( \mathscr{R}_z :=\big(zI - \mathbf{F}^{*}  H(u) \mathbf{F}\big)^{-1} \). Differentiating with respect to \( u \), we get
\begin{equation}
\begin{aligned}
\frac{\partial \mathcal{S}_\ell(u)}{\partial u} &= - \frac{1}{2}\frac{1}{2\pi i} \oint_{\mathcal{C}} \log^{+}(z)\cdot {\rm Trace} \left[ \mathscr{R}_z \mathbf{F}^{*} \frac{\partial H(u)}{\partial u} \mathbf{F} \mathscr{R}_z \right]_{\ell \ell}\,{\rm d}z.
 \end{aligned}
\end{equation}

The first and second variations of $\E(u)$ are defined as follows:
\begin{equation*}
\begin{aligned}
\langle \delta \E (u), v \rangle &=  \sum_{\ell \in \Lambda} \sum_{\rho \in \mathcal{R}} V_{,\rho}\big(Du(\ell)\big)  \cdot D_{\rho} v(\ell), \\
\langle \delta^2 \E (u)v, w \rangle &=  \sum_{\ell \in \Lambda} \sum_{\rho , \sigma \in \mathcal{R}}  V_{,\rho \sigma}\big(Du(\ell)\big) \cdot D_{\rho} v(\ell) D_{\sigma} w(\ell).
\end{aligned}
\end{equation*}
Hence, we can write
\begin{align*}
\langle \delta^2 \E (u) \textbf{F}^{\alpha}, \textbf{F}^{\beta} \rangle =  \sum_{\ell \in \Lambda} \sum_{\rho , \sigma \in \mathcal{R}}  V_{,\rho \sigma}\big(Du(\ell)\big) \cdot D_{\rho} \textbf{F}^{\alpha}(\ell) D_{\sigma} \textbf{F}^{\beta}(\ell).
\end{align*}

According to \eqref{th:properties_F}, we have
\begin{align}
\frac{\partial}{\partial u_n} \Big[ \textbf{F}^* H(u) \textbf{F}\Big]_{\alpha \beta}=  \sum_{\ell\in\Lambda} \sum_{\rho, \sigma, \tau \in \mathcal{R}} V_{, \rho\sigma\tau}\big(Du(\ell)\big) \cdot D_{\rho} \textbf{F}^{\alpha}(\ell) D_{\sigma} \textbf{F}^{\beta}(\ell) D_{\tau} e_n(\ell).
\end{align}
Hence, one can acquire 
% \cco{there are now $n$ on the LHS and in the sum, that doesn't work.} \cctina{maybe I should somehow mention the this is valid for $n$ in adjacency of $\ell$, like in the second some maybe I should add $n \approx \ell?$ cause obviously $D_{\tau} e_n(\ell)$ is only non zero near n.}
\begin{equation}
\begin{aligned}
\Big|\frac{\partial}{\partial u_n} \Big[ \textbf{F}^* H(u) \textbf{F}\Big]_{\alpha \beta}\Big| &\lesssim \sum_{\ell\in\Lambda} \Big(\sum_{\rho\in \mathcal{R}} \big|D_{\rho} \textbf{F}^{\alpha}(\ell) \big|^{2}\Big)^{1/2} \Big(\sum_{\sigma \in \mathcal{R}} \big|D_{\sigma} \textbf{F}^{\beta}(\ell)\big|^2\Big)^{1/2} \Big(\sum_{\tau\in \mathcal{R}} \big|D_{\tau} e_n(\ell)\big|^2\Big)^{1/2}  \\
&\lesssim \sum_{\ell \in \Lambda} \big|D F(\ell-\alpha)\big| \cdot \big|D F(\ell-\beta)\big|  \\
&\lesssim  |n - \alpha|^{-d} |n - \beta|^{-d}, 
\end{aligned}
\end{equation}
where we used the fact that $D_\tau e_n(\ell)$ is non-zero only for $\ell$ in a bounded neighbourhood of $\ell$, and the last inequality follows from Lemma~\ref{th:properties_F}. 
Applying Lemma ~\ref{th:resolvent estimate}, we have
\begin{equation}
\begin{aligned}
\left[\mathscr{R}_z \frac{\partial}{\partial u_n} \Big[ \textbf{F}^* H(u) \textbf{F}\Big]_{\alpha \beta} \mathscr{R}_z \right]_{\ell\ell}
& \lesssim \sum_{{\alpha, \beta} \in \Lambda} |\alpha -\ell|^{-d} | \beta - \ell|^{-d} |n - \alpha|^{-d} |n - \beta|^{-d} \\ 
& \lesssim \Big ( \sum_{{\alpha} \in \Lambda} |\alpha -\ell|^{-d}  |n - \alpha|^{-d} \Big)^2
 \\ 
& \lesssim \Big ( \int_1^{|n-\ell|} r^{-d} r^{d-1} |n- \ell|^{-d} \,{\rm d}r \Big )^2 + \int_{|n-\ell|}^{\infty} r^{d-1} r^{-2d} \,{\rm d}r  \\ 
& \lesssim |n -\ell|^{-2d}. 
\end{aligned}
\end{equation}

We can now estimate the derivative of the site entropy as
\begin{equation}
\begin{aligned}
\left| \frac{\partial \mathcal{S}_\ell(u) }{\partial u_n} \right| &\leq \frac{1}{2}\frac{1}{2\pi} \mathrm{Trace} \left| \oint_{\mathcal{C}} \log^+(z)\cdot \left[\mathscr{R}_z \mathbf{F}^{*} \frac{\partial H(u)}{\partial u_n} \mathbf{F} \mathscr{R}_z\ \right]_{\ell\ell} {\rm d}z \right| \\
&\lesssim \frac{1}{4\pi} \oint_{\mathcal{C}} \big|\log^+(z)\big| \cdot |n -\ell|^{-2d}\,{\rm d}z \\
&\leq C_2 |n -\ell|^{-2d},
\end{aligned}
\end{equation}
where $C_2>0$ is a constant. Let $r_{n\ell}:=|n - \ell|^{-2d}$. We can then establish the locality as stated in Theorem~\ref{thm:local}, thus completing the proof.
\end{proof}

\begin{proof}[Proof of Theorem~\ref{thm:err_local}]

For $u \in \dot{\mathcal{W}}^{1,2}$, we consider the Taylor expansion of $\widetilde{\mathcal{S}}^{+}_\ell$ at the reference homogeneous lattice, that is,
\begin{eqnarray}
    \widetilde{\mathcal{S}}_\ell(u) =  \widetilde{\mathcal{S}}_\ell({\bf 0}) + \sum_{m \in \Lambda} \frac{\partial \widetilde{\mathcal{S}}_\ell({\bf 0})}{\partial u_m} \cdot u_{m} + \text{higher order terms}.
\end{eqnarray}

According to the definition~\eqref{eq:truncated_def}, the difference in site entropy can be estimated by
\begin{align}
    \left|\mathcal{S}_\ell(u) - \widetilde{\mathcal{S}}_\ell(u) \right| &\leq \left| \sum_{|r_{\ell m}|>r_{\rm cut}} \frac{\partial \mathcal{S}_\ell({\bf 0})}{\partial u_{m}} u_m \right| \nonumber \\ %+ \mathcal{O}(\|u\|^2), \\
    &\leq \| u \|_{L^{\infty}} \cdot \sum_{|r_{\ell m}|>r_{\rm cut}} \left| \frac{\partial \mathcal{S}^{+}_\ell({\bf 0})}{\partial u_{m} } \right|. %+ \mathcal{O}(\|u\|^2).
\end{align}

Next, we employ Theorem \ref{thm:local} to estimate the decay of the derivatives of the site entropy. Specifically, we can write
\begin{align}
    \left|\mathcal{S}_\ell(u) - \widetilde{\mathcal{S}}_\ell(u) \right| &\leq \| u \|_{L^{\infty}} \cdot \sum_{|r_{\ell m}|>r_{\rm cut}} C_2 |r_{\ell m}|^{-2d} \nonumber \\
    &\leq C_3 \| u \|_{L^\infty} \cdot \int_{|r_{\ell m}| \leq r_{\rm cut}} |r_{\ell m}|^{-2d} \, \mathrm{d}V \nonumber \\
    &\leq C_3 \| u \|_{L^\infty} \cdot \int_{0}^{r_{\rm cut}} r^{-2d} \cdot r^{d-1} \, \mathrm{d}r,
\end{align}
where $\mathrm{d}V$ represents the volume element in the $d$-dimensional space. This can be estimated in terms of the cut-off radius $r_{\rm cut}$, resulting in the bound
\begin{align}
    \left|\mathcal{S}_\ell(u) - \widetilde{\mathcal{S}}_\ell(u) \right| & \leq C_3  \| u \|_{L^\infty} \cdot r_{\rm cut}^{-d},
\end{align}
where $C_3$ is a constant that depends on the volume of the integration domain and the constant $C_2$ from Theorem~\ref{thm:local}.
\end{proof}

\section{An alternative derivation of entropy}
\label{sec:projected_normal_modes}
In this section, we provide an alternative derivation of site entropy commonly used in physics and engineering \cite{FULTZ2010247, Dederichs1980}. Expanding on our discussion of lattice displacements from a classical standpoint in Section~\ref{subsec:FFE}, we now delve into a quantum mechanical perspective, assuming that a lattice vibration mode with frequency $\omega$ behaves like a simple harmonic oscillator, thereby being confined to specific energy values. We will show that this derivation  yields the same definition for entropy mentioned in Section \ref{subsec:FFE}. 

To that end, we first consider a harmonic oscillator with energy levels given by:
\begin{equation}
E_n = \hbar \omega \left( n + \frac{1}{2} \right) ,
\end{equation}
where \( \hbar \) is the reduced Planck constant, and \( \omega \) is the angular frequency.
The partition function for the quantum harmonic oscillator is the sum of the Boltzmann factors for all possible states~\cite{FULTZ2010247}, that is, 
\begin{align}
Z(\omega) = \sum_{n=0}^{\infty} e^{-\beta E_n} = e^{-\beta \hbar \omega/2} \cdot \sum_{n=0}^{\infty} e^{-\beta \hbar \omega n} = \frac{e^{-\beta \hbar \omega/2}}{1 - e^{-\beta \hbar \omega}},
\end{align}
where the last identity follows from the results of infinite geometric series and $\beta:=1/(k_{B}T)$.

Using the definition of Helmholtz free energy defined by \eqref{eq:our_def}, we can obtain
\begin{equation}
\begin{aligned} \label{eq:freeE}
\mathcal{A}(\omega) &= -\frac{1}{\beta}\log\big(Z(\omega)\big) \\
&= -\frac{1}{\beta} \left( -\frac{\beta \hbar \omega}{2} - \log(1 - e^{-\beta \hbar \omega}) \right) \\
&= \frac{1}{2} \hbar \omega + \frac{1}{\beta} \log(1 - e^{-\beta \hbar \omega}).
\end{aligned}
\end{equation}

Suppose that we have already obtained all the eigenstates $|\phi \rangle$ and eigenvalues $\omega_\alpha^2$ of the Hessian $H$, obtained through $(H - M\omega_\alpha^2) |\phi \rangle = |0 \rangle$. Then, we define the \textit{total density of states} (DOS) \cite{finnis03} by
\begin{align}\label{eq:DOS_def}
\Omega(\omega) := \sum_{\alpha = 1}^{3N} \delta(\omega - \omega_\alpha). 
\end{align}
The DOS can be comprehended in the operational sense that if we integrate $\Omega(\omega)$ over a frequency range $\omega_1$ to $\omega_2$, we therefore obtain the total number of states within that frequency range~\cite{Chen2018},  
\begin{align} \label{eq:DOS_QoI}
    \langle \Omega(\omega) , f \rangle = \int f(\omega) \Omega(\omega) \; \mathrm{d} \omega = \sum_{\alpha = 1}^{3N}  f(\omega_{\alpha}) \quad \text{for} \; f \in C(\R).
\end{align}

Hence, for a system of $N$ atoms with $3N$ degrees of freedom, we expect $\int_0^{\infty} \Omega(\omega) \,\mathrm{d} \omega  =3N$. Using \eqref{eq:DOS_QoI}, we can now determine the overall free energy of the system, $\mathcal{A}_{\text{total}}$, accounting for all vibrational modes as follows:
\begin{equation}\label{eq:Atotal}
\mathcal{A}_{\text{total}} = \sum_{\alpha = 1}^{3N} \mathcal{A}(w_\alpha) = \int_0^{\infty} \mathcal{A}(\omega) \Omega(\omega) \,\mathrm{d} \omega.
\end{equation}
Substituting \eqref{eq:freeE} into \eqref{eq:Atotal}, one can obtain
\begin{equation}
\mathcal{A}_{\text{total}} = \int_0^{\infty} \left( \frac{1}{2} \hbar \omega + k_B T \log(1 - e^{-\beta \hbar \omega}) \right) \Omega(\omega)\,\mathrm{d} \omega.
\end{equation}

Now, we can compute the derivative of $\mathcal{A}_{\text{total}}$ with respect to temperature to acquire the total entropy, that is,
\begin{equation}\label{eq:Stotal}
\mathcal{S}_{\text{total}} = k_B \int_0^{\infty} \left [\frac{\beta \hbar \omega}{e^{\beta \hbar \omega} - 1} - \log(1 - e^{-\beta \hbar \omega}) \right ] \Omega(\omega) \,\mathrm{d} \omega.
\end{equation}
In the classical limit of high temperatures, specifically when the temperatures exceed the crystal's Debye temperature such that \( \frac{\hbar \omega}{k_B T} \ll 1 \), \eqref{eq:Stotal} asymptotically approaches to
\begin{align} \label{eq:DOS_entropy}
\mathcal{S}_{\text{total}} &= k_B \int_0^\infty  \left[\log\left(\frac{kT}{\hbar\omega_\alpha}\right) + 1\right] \Omega(\omega) \,\mathrm{d} \omega.
\end{align}

To determine defect formation entropy, we need to assess the change in the total density of states (DOS) due to the presence of defect, i.e.,
\begin{align}
\mathcal{S}^{\rm defect} - S^{\rm ideal} = k_B \int_0^\infty  \left[\log\left(\frac{kT}{\hbar\omega_\alpha}\right) + 1\right] \Delta \Omega(\omega)  \,\mathrm{d} \omega,
\end{align}
where the term \( \Delta \Omega(\omega) \) represents the difference in DOS introduced by the defect, defined as \( \Delta \Omega(\omega) = \Omega^{0}(\omega) - \Omega(\omega) \). Here, $\Omega^{0}(\omega)$ and $\Omega(\omega)$ denote the total DOS of the defect and ideal lattice, respectively.  

In the case where the defect does not add additional degrees of freedom to the lattice, e.g. for a substitutional impurity, we have
$$\int_0^\infty \Delta \Omega(\omega)  \ \mathrm{d} \omega =0,$$
which leads to the fact that
\begin{align}
\Delta \mathcal{S} = \mathcal{S}^{\rm defect} - \mathcal{S}^{\rm ideal} = -k_B \int_0^\infty  \log (\omega) \Delta \Omega(\omega) \, \mathrm{d} \omega.
\end{align}

We can use the equivalent representations for ~\eqref{eq:DOS_def} and write
\begin{equation}
\begin{aligned}
\Delta \Omega(\omega) &= \sum_{\alpha} \big (\delta(\omega - \omega_\alpha) - \delta(\omega - \omega_{\alpha}^{0}) \big )
\\ &= 2 \omega \; \mathrm{Trace} \big ( \delta(\omega^2 - H) - \delta(\omega^2 - H_{0}) \big ).
\end{aligned}
\end{equation}
Thus, the defect formation entropy is
\begin{align}
\Delta \mathcal{S} = \frac{k_B}{2} \sum_{\alpha=1}^{3N} \; \log \frac{(\omega_\alpha^0)^2}{(\omega_\alpha)^2}
= \frac{k_B}{2} \; \mathrm{Trace}(\log H_0 - \log H)
= \frac{k_B}{2} \; \log (\frac{\mathrm{det}\; H_0}{\mathrm{det}\; H}), 
\end{align}
which is identical to our derivation ~\eqref{eq:our_def} shown in the main context.

Finally, we are in the stage of considering the spatial decomposition of $\Omega$ which would automatically lead to a spatial decomposition of $S$. Following the discussion in \cite{Dederichs1980}, we have
\begin{align} \label{eq:decomDOS}
    \Omega_\ell(\omega)= \sum_{\alpha} \delta(\omega - \omega_\alpha) [\phi_\alpha]_{\ell}^2,
\end{align}
where $[\phi_\alpha]_{\ell}$ is the $\ell$-th entry of $\phi_\alpha$. Using \eqref{eq:DOS_entropy} and \eqref{eq:decomDOS}, we can decompose the total entropy into local contributions from individual atoms,
\begin{align}
\mathcal{S}_{\ell}(\omega) := \langle \Omega_\ell(\omega), \xi	\rangle = \sum_{\alpha =1}^{3N}  \xi(\omega_\alpha)  [\phi_\alpha]_{\ell}^2 \quad \text{where} \  \xi := k_B \left[\log\left(\frac{kT}{\hbar\omega_\alpha}\right) + 1\right].
\end{align}
It is straightforward to verify that
\begin{align}
    \mathcal{S}_{\text{total}}= \sum_{\ell\in\Lambda} \mathcal{S}_\ell = k_B \sum_{\alpha =1}^{3N}  \left[\log\left(\frac{kT}{\hbar\omega_\alpha}\right) + 1\right]  [\phi_\alpha]_{\ell}^2. 
\end{align}

 We have shown that in the local basis, total entropy can be accurately partitioned into site-specific contributions. This local entropy is essential for understanding the thermodynamic properties of materials at the atomic level. Thus, the locality principle is not just a mathematical convenience but a reflection of the physical reality of how atomic vibrations contribute to the overall entropy of a system.

% Hence, we complete the alternative derivation of site entropy, which is commonly used in the realm of computational physics and engineering.

\section{Fast contour integration}
\label{sec:apd:contour}

We first recall that the definition of total entropy introduced in Section~\ref{sec:sub:site_entropy}, that is,
\begin{align}\label{eq:S_contour}
\mathcal{S}(u) :=& -\frac{1}{2}\operatorname{Trace}\,\log(\textbf{F}^*H(u)\textbf{F}) \nonumber \\
=& -\frac{1}{2} \frac{1}{2\pi i} \oint_{\mathcal{C}} \log (z)\cdot {\rm Trace} \big(  z\mathbf{I}-\pmb{\rm F}^{*}  H(u) \pmb{\rm F}\big)^{-1}\,{\rm d}z,
\end{align}
where the second identity follows from the definition of the logarithm of an operator. To compute~\eqref{eq:S_contour} effectively, it is necessary to employ numerical integration methods for contour integration. However, conventional numerical integration approaches may become inefficient, particularly when dealing with the spectrum of the operator of interest, denoted as $\pmb{\rm F}^{*} H(u) \pmb{\rm F}$ in our context, which does not lend itself to a straightforward integration path.

To overcome this challenge, we integrate the trapezoidal rule with conformal maps that incorporate Jacobi elliptic functions, as suggested in \cite{doi:10.1137/070700607}. We employ a conformal mapping technique that relocates the contour to a more uniform domain. This allows us to perform the contour integral along a path that avoids the branch cut of the logarithm. The transformed contour is discretized, and the trapezoidal rule is applied to approximate the integral, resulting in a substantial improvement in computational efficiency.

More precisely, we first apply a change of variables that could achieve the improved convergence rate based on the discussion in~\cite{doi:10.1137/070700607}.  Introducing a new variable $w = \sqrt{z}$, ${\rm d}z=2w\,{\rm d}w$, we can express \eqref{eq:S_contour} as
\begin{eqnarray}\label{eq:new_Su}
\mathcal{S}(u) = - \frac{1}{2\pi i} \oint_{\mathcal{C}_w} w \log (w^2)\cdot {\rm Trace} \big(  w^2\mathbf{I}-\pmb{\rm F}^{*}  H(u) \pmb{\rm F}\big)^{-1}\,{\rm d}w.
\end{eqnarray}

The conformal mapping technique involves multiple transformations aimed at mapping the region of analyticity of $f$ and $\big(w^2\mathbf{I}-\pmb{\rm F}^{*}  H(u) \pmb{\rm F}\big)^{-1}$, which is the doubly connected set $\Xi = \mathbb{C} \setminus  \big( (-\infty,0] \cup [\frac{m}{2}, \frac{M}{2}]\big)$, to an annulus $A=\{z \in \mathbb{C}: r<|z|<R \}$, where $r$ and $R$ represent the inner and outer radii of the annulus, respectively.

We first map the annulus to a rectangle with vertices $\pm K$ and $\pm K+iK'$, using a logarithmic transformation
\begin{equation}
t(s) = \frac{2Ki}{\pi} \log\left(-\frac{is}{r}\right),
\end{equation}
where $K, K'$ denote the complete elliptic integrals. Next, the rectangle is mapped to the upper half-plane in the $u$-plane by the Jacobian elliptic function
\begin{equation}
u(t) = \text{sn}(t|k^2), \quad k = \frac{(M/m)^{1/4} - 1}{(M/m)^{1/4} + 1},
\end{equation}
where $M$ and $m$ correspond to the maximum and minimum eigenvalues of the matrix $\pmb{\rm F}^{*}  H(u) \pmb{\rm F}$.

A Möbius transformation is then applied to map the upper half-plane to the $z$-plane
\begin{equation}
w(u) = (M/m)^{1/4} \left(\frac{k^{-1} +u}{k^{-1} - u} \right).
\end{equation}
This final transformation is designed to distribute the eigenvalues of $A$ evenly along the real axis, thus facilitating the application of the trapezoidal rule~\cite{doi:10.1137/070700607}. 

Taking into account the aforementioned transformations with \eqref{eq:new_Su}, we can rewrite \eqref{eq:S_contour} as
\begin{eqnarray}\label{eq:S_contour_cint}
\mathcal{S}(u) = - \frac{1}{2\pi i} \oint_{-K + iK'/2}^{3K + iK'/2} w \log (w^2(t))\cdot {\rm Trace} \big(  w^2(t)\mathbf{I}-\pmb{\rm F}^{*}  H(u) \pmb{\rm F}\big)^{-1}\cdot \frac{{\rm d}w}{{\rm d}u} \cdot \frac{{\rm d}u}{{\rm d}t}\,{\rm d}t.
\end{eqnarray}
Applying the trapezoid rule with $N$ equally spaced points on the region $(-K+iK/2,K+iK/2)$, we can write
\begin{eqnarray}
\mathcal{S}(u) = \frac{4K(mM)^{1/4}}{\pi Nk} {\rm Im}\left(\sum_{j=1}^{N} \frac{w(t_j) f(w(t_j)^2)\big(w(t_j)^2 \mathbf{I} - \pmb{\rm F}^{*}  H(u) \pmb{\rm F}\big)^{-1} {\rm cn}(t_j) {\rm dn}(t_j)}{\big(k^{-1} - u(t_j)\big)^2}\right),
\end{eqnarray}
where 
\begin{equation*}
t_j = -K + \frac{iK'}{2} + 2 \frac{(j - \frac{1}{2})K}{N}, \quad 1 \leq j \leq N.
\end{equation*}

For the evaluation of the derivative of site entropy defined by \eqref{eq:local_S+}, we can use the same approach to obtain
\begin{align}\label{eq:partial}
\frac{\partial S_\ell}{\partial u_n} &=  -\frac{1}{2}  \frac{1}{2\pi i} \oint_{\mathcal{C}} \log (z)\cdot (\mathscr{R}_z\pmb{\rm F}_{\ell})^{*} \frac{\partial H(u)}{\partial u_n} (\mathscr{R}_z\pmb{\rm F}_{\ell})\,{\rm d}z \nonumber \\
&= \frac{4K(mM)^{1/4}}{\pi Nk} {\rm Im}\left(\sum_{j=1}^{N} \frac{w(t_j) f(w(t_j)^2) (\mathscr{R}_z\pmb{\rm F}_{\ell})^{*} \frac{\partial H(u)}{\partial u_n} (\mathscr{R}_z\pmb{\rm F}_{\ell})
{\rm cn}(t_j) {\rm dn}(t_j)}{\big(k^{-1} - u(t_j)\big)^2}\right),
\end{align}
where $\pmb{\rm F}_{\ell}$ represents the $\ell$-th column of $\pmb{\rm F}$ with $\ell\in\L$.

It is worthwhile mentioning that an additional computational cost for evaluating \eqref{eq:partial} stems from the computation of $\partial H(u)/\partial u_n$. In the locality test shown in Section~\ref{sec:th_rs}, one requires to compute $\partial H(u)/\partial u_n$ with respect to all sites $n$, which results in computing the Jacobian of the Hessian matrix.
We leverage the sparsity of the Hessian matrix to compute its Jacobian tensor. The sparsity pattern of a differential matrix in advance can greatly streamline computations. This is achieved by compressing the sparse matrix into a denser format, allowing for the processing of non-zero entries with significantly fewer function calls than previously necessary. 
This technique uses a strategy from graph theory \cite{SparseJac, kubale2004graph, wang2018posteriori}, aiming to combine columns with non-overlapping non-zero elements, referred to as {\it structurally orthogonal columns} into single groups, thus reducing the total number of groups needed. 

Once a group is determined, centered finite difference or automatic differentiation \cite{10.5555/3122009.3242010} can be used to calculate the directional derivatives along the compressed matrix directions. For more details on the graph coloring method for computing derivatives we refer to~\cite{SparseJac}. We developed and employed the {\tt ComplexElliptic.jl} \cite{gitComplexElliptic} package to perform conformal maps. Additionally, we utilized the \texttt{SparsityDetection.jl} \cite{gowda2019sparsity}, \texttt{Symbolics.jl} \cite{gowda2021high} and \texttt{ForwardDiff.jl} \cite{RevelsLubinPapamarkou2016} packages to evaluate the sparsity patterns of the Jacobians.

\section{Atomic Cluster Expansion}
\label{sec:apd:ACE}
We provide a brief overview of the Atomic Cluster Expansion (ACE) potential and refer to~\cite{2019-ship1, 2021-qmmm3, wang2020posteriori, witt2023otentials} for more detailed construction and discussion. The ACE site potential can be systematically formulated as an atomic body-order expansion for a given correlation order $N \in \mathbb{N}$,
\begin{equation}
V_{\text{ACE}}\big(\{\mathbf{y}_j\}_{j=1}^J\big) = \sum_{N=0}^{\mathscr{N}} \frac{1}{N!} \sum_{j_1, \dots, j_N = 1}^J V_N(\mathbf{y}_{j_1}, \ldots, \mathbf{y}_{j_N}),
\end{equation}
where the $N$-body potential $V_N : \mathbb{R}^{dN} \rightarrow \mathbb{R}$ is approximable by using a tensor product basis~\cite{2019-ship1},
\begin{equation*}
\phi_{{\bm n} {\bm \ell} {\bm m}}\big(\{\mathbf{y}_j\}_{j=1}^N\big) := \prod_{j=1}^{N} \phi_{{{\bm n}_j {\bm \ell}_j {\bm m}_j}}(\mathbf{y}_j),
\end{equation*}
where $\phi_{{\bm n} {\bm \ell} {\bm m}}(\mathbf{y}) := P_n(y)Y_{{\bm \ell}}^{{\bm m}}(\hat{\mathbf{y}}),$ with $\mathbf{y} \in \mathbb{R}^d$, $y = |\mathbf{y}|$, and $\hat{\mathbf{y}} = \mathbf{y}/y$. The functions $Y_{\ell}^{m}$ denote the complex spherical harmonics for $\ell = 0,1,\ldots$, and $m=-\ell,\ldots,\ell$, and $P_n(r)$ are radial basis functions for $n=0,1,\ldots$. 

This parameterization is already invariant under permutations (or, relabelling) of the input structure $\{{\bf y}_j\}_j$. One then symmetrizes the tensor product basis with respect to the group $O(3)$, employing the representation of that group in the spherical harmonics basis. This results in a linear parameterization
\begin{align*}
    V_{\text{ACE}}\big(\{\mathbf{y}_j\}_{j=1}^J\big)
    &= \sum_{{{\bm n} {\bm \ell} q}} c_{{{\bm n} {\bm \ell} q}} B_{{\bm n} {\bm \ell} q} (\{\mathbf{y}_j\}), \\ 
    \text{where} 
    \quad     
    B_{{\bm n} {\bm \ell} q} (\{\mathbf{y}_j\})
    &= \sum_{{\bm m } \in \mathcal{M}_{ {\bm \ell}}} \mathcal{C}^{{\bm n} {\bm \ell} q}_{{\bm m }} {\bm A}_{{\bm n} {\bm \ell}{\bm m} }(\{\mathbf{y}_j\}), \\ 
    \text{and} 
    \quad 
    {\bm A}_{{\bm n} {\bm \ell}{\bm m} }(\{\mathbf{y}_j\}) &= \prod_{\alpha=1}^{N} \sum_{j=1}^{J} \phi_{{\bm n}_{\alpha}{\bm \ell}_{\alpha}{\bm m}_{\alpha}}(\mathbf{y}_j).
\end{align*}
The $q$-index ranges from $1$ to $n_{\bm n \bm \ell}$ and enumerates the number of all possible invariant couplings through the generalized Clebsch--Gordan coefficients $\mathcal{C}^{{\bm n} {\bm \ell} q}_{{\bm m }}$. The 
This parameterization was originally devised in \cite{PhysRevB.99.014104}. Its approximation properties and computational complexity are analyzed in detail in \cite{2019-ship1, 2021-apxsym, PhysRevMaterials.6.013804, braun2022higher}. The implementation we employ in our work is described in \cite{witt2023otentials}. The latter reference also described the details of the choice of radial basis $P_n$.

To complete the description of our parameterization, we select a finite subset of the basis, 
\begin{align*}
\mathbf{B} := \Big\{ B_{{\bm n} {\bm \ell} q} \,\Big|\, & ({\bm n}, {\bm \ell}) \in \mathbb{N}^{2N} \text{ ordered}, ~ \sum_{\alpha} \ell_{\alpha} \text{ even},  
    \sum_{\alpha} m_\alpha = 0, \\ 
    & q = 1, \ldots, n_{{\bm n}{\bm \ell}}, 
    \sum_\alpha \ell_\alpha + n_\alpha \leq D_{\rm tot}, 
    N \leq \mathscr{N} 
    \Big\},
\end{align*}
where the two approximation parameters are the correlation order $\mathscr{N}$ and the total degree $D_{\rm tot}$. With this selection of the basis we can write our ACE parameterization more convenient as 
%
% We can now describe the resulting basis by
% where $n_{{\bm n}{\bm \ell}}$ is the number of basis functions for the chosen $({\bm n}, {\bm \ell})$~\cite{2019-ship1}. After building the finite symmetric polynomial basis set $ \mathbf{B} \subset \bigcup_{N=1}^{\mathscr{N}} \mathbf{B}_N $, every choice of parameters $\mathbf{c} = (c_B)_{B \in \mathbf{B}}$ defines a site potential
\begin{equation}
V^{\text{ACE}}(\{\mathbf{y}_j \}; \mathbf{c}) = \sum_{B \in \mathbf{B}} c_B B(\{\mathbf{y}_j\}), 
\end{equation}
where ${\bf c} = (c_B)_{B \in {\bf B}}$ are the model parameters. 
%
% The force for this potential is represented by $F^{\text{ACE}}$. 
Due to the completeness of this representation, as the approximation parameters (like body-order, cut-off radius, and expansion precision) approach infinity, the model can represent any arbitrary potential \cite{2019-ship1}.

\label{sec:apd:ACE}

\bibliographystyle{plain}
\bibliography{bib.bib}

\end{document}

%% file: notation.tex
% formatting in math environments
% all notations should be defined here.

\renewcommand\arraystretch{1.5}

\def\per{{\rm per}}
\def\mA{{\sf A}}
\def\T{\mathcal{T}}
\def\D{\mathcal{D}}
\def\Z{\mathbb{Z}}
\def\a{\rm a}
\def\asP{\textnormal{\bf (P)}}
\def\asD{\textnormal{\bf (D)}}
\def\asRC{\textnormal{\bf (RC)}}
\def\asSE{\textnormal{\bf (S)}}
\def\asSER{\textnormal{\bf (S.R)}}
\def\assERL{\textnormal{\bf (RL)}}
\def\asSEL{\textnormal{\bf (S.L)}}
\def\asSEH{\textnormal{\bf (S.H)}}
\def\asSEP{\textnormal{\bf (S.P)}}
\def\asSEI{\textnormal{\bf (S.I)}}
\def\asSEPS{\textnormal{\bf (S.PS)}}
\def\asLS{\textnormal{\bf (LS)}}
\def\asS{\textnormal{\bf (S)}}
\def\pt{\textnormal{\bf (PT)}}
\def\VMM{V^{\rm ACE}}
\def\VQM{V^{\rm QM}}
\def\Vta{V^{\rm Taylor}}
\def\Vhom{V^{\rm h}}
\def\Ih{I_1^{\rm h}}
\def\Wb{\mathbf{W}}
\def\Vsh{V^{\rm SHIPs}}
\def\E{\mathcal{E}}
\def\EH{\E^{\rm H}}
\def\uh{u_{\rm h}}
\def\vh{v_{\rm h}}
\def\p0{\pmb{0}}
\def\Egfc{\E^{\rm GFC}}
\def\EHT{\widetilde{\E}^{\rm H}}
\def\Vb{V^{\rm BUF}_{\#}}
\def\Usx{\mathscr{U}^{\rm H}}
\def\Us{\mathscr{U}}
\def\uH{\bar{u}^{\rm H}}
\def\UsH{{\mathscr{U}}^{1,2}}
\def\THu{T^{\rm H}\bar{u}}
\def\lgamma{\ell^2_{\gamma}}
\def\efit{\varepsilon^{\rm E}}
\def\ffit{\varepsilon^{\rm F}}
\def\fcfit{\varepsilon^{\rm FC}_{\rm hom}}
\def\vfit{\varepsilon^{\rm V}}
\def\L{\Lambda}
\def\R{\mathbb{R}}
\def\Rl{\mathcal{R}}
\def\Rg{\mathcal{R}^{*}}
\def\Adm{{\rm Adm}}
\def\<{\langle}
\renewcommand{\>}{\rangle}
\def\LQM{\Lambda^{\rm QM}}
\def\LMM{\Lambda^{\rm MM}}
\def\LFF{\Lambda^{\rm FF}}
\def\Lbuf{\Lambda^{\rm BUF}}
\def\Lhom{\L^{\rm h}}
\def\LhomS{\L^{\rm h}_*}
\def\OQM{\Omega^{\rm QM}}
\def\OMM{\Omega^{\rm MM}}
\def\OFF{\Omega^{\rm FF}}
\def\Obuf{\Omega^{\rm BUF}}
\def\RQM{R_{\rm QM}}
\def\RMM{R_{\rm MM}}
\def\RFF{R_{\rm FF}}
\def\Rbuf{R_{\rm BUF}}
\def\rcut{R_{\rm c}}
\def\Rcore{R_{\rm DEF}}
\def\FMM{\F^{\rm ACE}}
\def\FQM{\F^{\rm QM}}
\def\F{\mathcal{F}}
\def\bF{\textbf{F}}
\def\N{\mathbb{N}}
\def\FH{\F^{\rm H}}
\def\Fa{\widetilde{\F}}
\def\Ea{\widetilde{\E}}
\def\Adm{{\rm Adm}}
\def\Admu{\mathscr{A}}
\def\wf{\mathfrak{w}}
\def\Hw{\mathscr{L}}
\def\fc{f_{\rm c}}
\def\Rc{R_{\rm c}}
\def\dx{\,{\rm d}x}
\def\dt{\,{\rm d}t}
\def\ds{\,{\rm d}s}
\def\Wcb{W_{\rm cb}}
\def\A{\mathbf{A}}
\def\n{\mathfrak{n}}